\documentclass[11pt,a4paper]{article}
\usepackage[utf8]{inputenc}
\usepackage{amsmath,amssymb,amsfonts,setspace,url,color,hyperref}
\usepackage{fullpage}
\usepackage{latexsym}
\usepackage{amsthm}
\usepackage{subfig}
\usepackage{float}
\usepackage{algorithmic}
\usepackage[ruled]{algorithm}
\usepackage{bigstrut,array,multirow}
\usepackage{tikz}
\usetikzlibrary{positioning,decorations.pathreplacing}
\theoremstyle{plain}
\newtheorem{theorem}{Theorem}
\newtheorem{lemma}{Lemma}

\newtheorem{definition}{Definition}

\newtheorem{proposition}{Proposition}
\newtheorem{claim}{Claim}

\theoremstyle{definition}

\newcommand{\Bo}{\mathrm{Box}}
\newcommand{\bord}{\partial}
\newcommand{\Uu}{\mathcal{U}}
\newcommand{\Vv}{\mathcal{V}}
\newcommand{\Ww}{\mathcal{W}}

\newcommand{\Gg}{\mathcal{G}}
\newcommand{\Cc}{\mathcal{C}}

\newcommand{\Ramp}{\mathcal{R}}
\newcommand{\Rgen}{\mathfrak{R}}

\newcommand{\Nn}{\mathcal{N}}

\newcommand{\ceil}[1]{\left\lceil #1 \right\rceil}
\newcommand{\ket}[1]{| #1 \rangle}
\newcommand{\floor}[1]{\left\lfloor #1 \right\rfloor}

\usepackage{tikz}
\usepackage{verbatim}
\usepackage{tkz-graph}
\usetikzlibrary{positioning,decorations.pathreplacing}
\newcommand{\Pp}{\mathcal{P}}

\begin{document}
\title{Average-Case Quantum Advantage with Shallow Circuits}
\author{
Fran{\c c}ois Le Gall\\
Graduate School of Informatics\\
Kyoto University\\
\url{legall@i.kyoto-u.ac.jp}
}
\date{}

\maketitle
\thispagestyle{empty}
\setcounter{page}{0}
\begin{abstract}
Recently Bravyi, Gosset and K\"onig (Science 2018) proved an unconditional separation between the computational powers of small-depth quantum and classical circuits for a relation. In this paper we show a similar separation in the average-case setting that gives stronger evidence of the superiority of small-depth quantum computation: we construct a computational task that can be solved on all inputs by a quantum circuit of constant depth with bounded-fanin gates (a ``shallow'' quantum circuit) and show that any classical circuit with bounded-fanin gates solving this problem on a non-negligible fraction of the inputs must have logarithmic depth. Our results are obtained by introducing a technique to create quantum states exhibiting global quantum correlations from any graph, via a construction that we call the \emph{extended graph}.

Similar results have been very recently (and independently) obtained by Coudron, Stark and Vidick (arXiv:1810.04233), and Bene Watts, Kothari, Schaeffer and Tal (STOC 2019).
\end{abstract}
\newpage
\section{Introduction}
\subsection{Background and our results}
A fundamental problem in quantum complexity theory is to prove the superiority of quantum computation over classical computation. While this has been shown in constrained models of computation such as query complexity (see for instance \cite{Ambainis18} for a recent survey), in weak models of computation like finite-state automata \cite{Kondacs+FOCS97}, and when considering relativized complexity classes (see, e.g., \cite{Bernstein+SICOMP97} for the first results and \cite{Raz+18} for the most recent breakthrough), no definite answer is known in standard computational models such as Turing machines or general circuits. Indeed, since the complexity class $\mathsf{BQP}$ corresponding to the problems that can be solved efficiently by a quantum computer satisfies the inclusions $\mathsf{P}\subseteq \mathsf{BQP}\subseteq \mathsf{PSPACE}$, unconditionally separating $\mathsf{P}$ and $\mathsf{BQP}$ cannot be shown without separating $\mathsf{P}$ and $\mathsf{PSPACE}$.

A recent active research area focuses on conditionally showing the superiority of quantum computation. Under several assumptions from computational complexity such as non-collapse of the polynomial hierarchy, the superiority of quantum computation with respect to classical computation has been shown in the standard circuit model in the worst-case setting \cite{Aaronson+STOC11,Aaronson+A14,Aaronson+CCC17,Bremner+10,Fahri+16,Fujii+16,Fujii+PRL18,Morimae+PRL14,Terhal+04} and even in the average-case setting \cite{Aaronson+STOC11,Aaronson+A14,Aaronson+CCC17,Bouland+18,Bremner+PRL16,Bremner+17,Fujii+16}. Note that showing the superiority in the average-case setting is a much stronger evidence of the superiority of quantum computation than a proof for the worst-case setting.

A recent breakthrough by Bravyi, Gosset and K\"onig \cite{Bravyi+17,Bravyi+18} showed an \emph{unconditional} separation between the computational powers of quantum and classical small-depth circuits: they constructed a computational problem that can be solved by quantum circuits of constant depth with bounded-fanin\footnote{In this paper the term \emph{bounded-fanin} means, as usual, that the fanin is bounded from above by a constant.} gates (``shallow quantum circuits'') and showed that any classical circuit with bounded fanin gates solving this problem on all inputs must have depth $\Omega(\log m)$, where $m$ denotes the input size. Besides being the first such unconditional separation in the circuit model, this separation is also especially important since shallow quantum circuits are likely to be the easiest quantum circuits to experimentally implement, due to their robustness to noise and decoherence. (Note that separations were already known when allowing gates with unbounded fanin or fanout \cite{Green+02, Hoyer+TOC05,Takahashi+CCC16}. The strength of Bravyi, Gosset and K\"onig's result is that it holds for the weaker model of quantum circuits with bounded fanin and fanout.)

The original classical lower bound shown in \cite{Bravyi+17} required the classical circuit to output the correct answer (with high probability) on each input, i.e., this was only a worst-case hardness result. Showing the advantages of shallow quantum circuits for a distribution (i.e., proving a corresponding average-case hardness result), which would give a significantly stronger evidence of the advantage of quantum shallow circuits, was discussed in \cite[Section 5]{Bravyi+17} and referred to as a ``challenging open question''. The recently published journal version \cite{Bravyi+18} partially answers this open question: it presents an average-case lower bound showing that any classical circuit that outputs the correct answer on a constant fraction of some restricted subset of the inputs (which can be efficiently sampled) must have logarithmic depth. In other words, it shows that any sublogarithmic-depth classical algorithm will fail with some constant probability on a input chosen uniformly at random in this restricted subset. 

In this work we give a stronger average-case hardness result. Our main result is the following theorem.
\begin{theorem}\label{th:main-amp}
There exists a relation $\Ramp\subseteq  \{0,1\}^{M}\times \{0,1\}^N$ for which the following two assertions hold.\footnote{As usual in computational complexity, the subset $\Ramp\subseteq \{0,1\}^M\times \{0,1\}^N$ is interpreted as the following computational problem: given an input $x\in\{0,1\}^M$, output any element of the set $\{z\in\{0,1\}^N\:|\:(x,z)\in\Ramp\}$. Through this paper we will use the convenient notation $\Ramp(x)=\{z\in\{0,1\}^N\:|\:(x,z)\in\Ramp\}$, for any $x\in\{0,1\}^N$.}
\begin{itemize}
\item
There is a constant-depth quantum circuit with bounded-fanin gates (i.e., a shallow quantum circuit) that on any input $x\in\{0,1\}^{M}$ outputs an element in the set $\Ramp(x)$ with probability 1. 
\item
There is a constant $\gamma>0$ such that any randomized circuit $C$ with bounded-fanin gates satisfying
\[
\frac{1}{2^{M}}\sum_{x\in\{0,1\}^{M}}\Pr[C(x)\in \Ramp(x)]\ge \frac{1}{\exp(\gamma \sqrt{M})}
\]
has depth $\Omega(\log M)$.
\end{itemize}
\end{theorem}
Theorem \ref{th:main-amp} thus shows the existence of a computational problem that can be solved by a shallow quantum circuit on all inputs but such that any classical circuit with bounded-fanin gates solving this problem on a non-negligible fraction of the inputs must have logarithmic depth. This gives an average-case result that is a strengthening of the average-case result from \cite{Bravyi+18} with respect to two aspects. First, our lower bound holds for any classical circuit that solves the problem on a non-negligible  fraction of the inputs (even exponentially small), and not only on a constant fraction. Second, our statement does not make any restriction on the set of inputs for which the hardness is established, i.e., it shows that any sublogarithmic-depth classical algorithm will fail with high probability on an input chosen uniformly at random in the whole set $\{0,1\}^M$.

\subsection{Overview of our techniques}
\paragraph{Main technical result.}
Our central technical result is the following theorem.
\begin{theorem}\label{th:main}
There exists a relation $R\subseteq \{0,1\}^m\times \{0,1\}^n$ for which the following two assertions hold.
\begin{itemize}
\item
There is a constant-depth quantum circuit with bounded-fanin gates (i.e., a shallow quantum circuit) that on any input $x\in\{0,1\}^m$ outputs a string in the set $R(x)$ with probability 1. 
\item
There is a constant $\alpha>0$ such that any randomized circuit $C$ with bounded-fanin gates satisfying
\[
\frac{1}{2^m}\sum_{x\in\{0,1\}^m}\Pr[C(x)\in R(x)]\ge 1-\alpha
\]
has depth $\Omega(\log m)$.
\end{itemize}
\end{theorem}
Theorem \ref{th:main-amp} is obtained from Theorem \ref{th:main} by amplifying the soundness using standard techniques: the relation $\Ramp$ is obtained by taking the direct product of $t$ copies of the relation $R$ for some sufficiently large integer $t$ (the sizes of the inputs and outputs in $\Ramp$ are thus $M=mt$ and $N=nt$). We show in Section \ref{sec:amp} how the soundness is then amplified from $1-\alpha$ to $(1-\alpha)^{t'}$ with $t'\approx t$ by this process and observe that $(1-\alpha)^{t'}$ is upper bounded by $1/\exp(\gamma \sqrt{M})$ for some constant $\gamma$. Note that this approach can also be applied to amplify the soundness of the average-case result in \cite{Bravyi+18}, which directly gives a result similar to Theorem \ref{th:main-amp} (but for a hard distribution that is not simply the uniform distribution).

\paragraph{Techniques from prior works.}
Before presenting our techniques we first describe how the result from \cite{Bravyi+17,Bravyi+18} was obtained. A central technical tool is a simple but fascinating result by Barrett et al.~\cite{Barrett+PRA07} that shows that correlations arising from local entanglement cannot be simulated classically without global interaction. This result was also used recently to show a separation between quantum and classical distributed computing \cite{LeGall+ArXiv18}. More precisely, \cite{Barrett+PRA07} considers the problem of simulating the measurement outcomes that occur when measuring each qubit of a well-chosen quantum state on $n$ qubits (the graph state associated with a cycle of length $n$) in either the $X$-basis or the $Y$-basis (the choice of the basis depends on input bits), and shows that creating the resulting output distribution classically requires coordinating the outcomes of qubits located at distance $\Omega(n)$ on the cycle. This result can actually easily be adapted to show that any classical circuit with one-dimensional nearest-neighbor architecture and bounded-fanin gates requires logarithmic depth to create this distribution, since otherwise distant wires cannot interact. Since a graph state over a cycle (and more generally over any constant-degree graph) can be created using a shallow quantum circuit, this already gives an unconditional separation between the computational power of quantum shallow circuits and the computational power of this restricted class of small-depth classical circuits.

The main contribution of \cite{Bravyi+17,Bravyi+18} is to show how to get a similar separation without restricting the topology of the classical circuit (other than its depth, naturally). A first important observation is that while 
interactions can now naturally occur between distant wires, any sublogarithmic-depth bounded-fanin classical circuit~$C$ cannot create interactions between all pairs of wires. Ref.~\cite{Bravyi+17} showed that it is then always possible to find a large subset of wires $S_C$ that are connected as a long cycle and in which distant wires do not interact. The key idea is then to consider a computational problem (called 2D Hidden Linear Function) where the input is divided in two parts: one part specifies the basis in which the qubits of the graph state are measured and the second part the topology of the graph state. By using the second part of the input to force the graph state to use only nodes corresponding to wires in $S_C$, the same argument as in \cite{Barrett+PRA07} can be again applied on the cycle defined by $S_C$ to conclude that the sublogarithmic-depth classical circuit $C$ cannot output a valid output with high probability. 

\paragraph{Our approach.}
Let us now describe the main ideas of our approach to prove Theorem \ref{th:main}. Our main technical tool, described in Section~\ref{sec:edge-graph} is a generalization of the construction from \cite{Barrett+PRA07}: we show how to generate useful quantum correlations not only from a cycle but also from any undirected graph~$G$. The key insight is to consider what we call the \emph{extended graph of} $G$, denoted $\overline G$, which is obtained by adding a vertex on each edge of~$G$. We show that when measuring the qubits of the graph state corresponding to $\overline{G}$ in either the $X$-basis or the~$Y$-basis, we get probability distributions that satisfy global conditions related to properties of subgraphs (in particular paths and cycles) of $\overline{G}$. The conditions are described in Theorems~\ref{th:cond1} and~\ref{th:cond2}.  

In order to prove our separations we consider a $d^3\times d^3$ square grid in which one vertex (called a control vertex) is placed at the center of each $1\times 1$ square of the grid (and connected by 4 edges to the~4 corners of the square), and then adding one vertex on each edge. The final graph is denoted $\overline{\Gg}_d$. The construction is described in Section \ref{sec:construction}. Note that by construction~$\overline{\Gg}_d$ is an extended graph. This means that the probability distributions arising when measuring the qubits of the graph state associated with this graph, which we denote $\ket{\overline{\Gg}_d}$, can be described by Theorems \ref{th:cond1} and~\ref{th:cond2}.

We can now describe the computational problem that we consider to show our separation. Let $m$ denote the number of control vertices in $\overline{\Gg}_d$ and $n$ denote the total number of vertices. Observe that $m=\Theta(d^6)$ and $n=\Theta(d^6)$. Given as input a string of bits $x\in \{0,1\}^m$, we consider the following process: measure each qubit of the quantum state $\ket{\overline{\Gg}_d}$ in the $X$-basis except the qubits corresponding to the control vertices, which are measured either in the $X$-basis or in the $Y$-basis depending on the value of $x$. The relation $R$ considered to prove Theorem \ref{th:main} simply asks, given $x\in\{0,1\}^m$ as input, to compute any sequence of measurement outcomes $z\in\{0,1\}^n$ that has non-zero probability of being obtained by this process. 
Note that this problem can be solved by shallow quantum circuits: the graph $\overline{\Gg}_d$ has constant degree and thus the graph state $\ket{\overline{\Gg}_d}$ can be constructed in constant depth. 

In Section \ref{sec:proof} we first show that for any sublogarithmic-depth bounded-fanin classical circuit $C$ there exists a subset $S_C$ of wires that are connected as a long cycle and in which distant wires do not interact. The proof of this claim is similar to what was done in \cite{Bravyi+17,Bravyi+18}.
We then show that this claim, along with Theorems \ref{th:cond1} and \ref{th:cond2}, are enough to prove that the sublogarithmic-depth classical circuit $C$ cannot output a valid output with high probability. The key point of our argument -- and the reason why our result holds for average-case hardness on the whole set $\{0,1\}^m$ of possible inputs and not only for worst-case hardness or average-case hardness on a restricted set of inputs -- is that we do not need to construct the graph state corresponding to the subgraph induced by~$S_C$, i.e., we do not need to adapt the topology of the measured graph state to the circuit. Theorems \ref{th:cond1} and \ref{th:cond2} guarantee that we can instead work with the graph state~$\ket{\overline{\Gg}_d}$ corresponding to the whole graph and simply look at the relevant part of the probability distribution (the part corresponding to the wires in $S_C$). 

\paragraph{Related works.}
A similar result has been recently (and independently) obtained by Coudron, Stark and Vidick and expanded into a framework for robust randomness expansion \cite{Coudron+18}. The proof techniques are nevertheless different: \cite{Coudron+18} constructs a problem hard for small-depth classical circuits by starting with a non-local game and showing how to plant a polynomial number of copies of the game into a graph. Our approach, on the other hand, starts with a graph and shows how to create from it a quantum state exhibiting global quantum correlations that cannot be simulated by small-depth classical circuits with bounded-fanin gates.

An even stronger result has been very recently announced: Bene Watts, Kothari, Schaeffer and Tal \cite{Bene+18} have shown that the 2D Hidden Linear Function introduced in \cite{Bravyi+17,Bravyi+18} cannot be solved on a non-negligible fraction of the inputs even by small-depth classical circuits with unbounded-fanin parity gates.

\section{Preliminaries}\label{sec:prelim}
\subsection{General notations and a technical lemma}
Given a Boolean function $f\colon A\to \{0,1\}$ on a finite set $A$, we write $|f|$ the number of elements $a\in A$ such that $f(a)=1$, i.e., $|f|=\sum_{a\in A} f(a)$. Similarly, for any finite binary string $x\in\{0,1\}^\ast$, we denote $|x|$ the Hamming weight of $x$, i.e., the number of non-zero bits of $x$.

All the graphs considered in this paper will be undirected. Given a graph $G=(V,E)$ and any vertex $u\in V$, we denote 
\[
\Nn(u)=\{v\in V\:|\:\{u,v\}\in E\}
\]
the set of neighbors of $u$. Given a path $p$ in the graph $G$ we will often be mainly interested only in the set of vertices on the path. For a vertex $v\in V$, we will thus use the convenient notation $v\in p$ to express the fact that $v$ is on the path $p$. 

The notation $\oplus$ will denote the addition modulo 2 (i.e., the bit parity). We will use the following lemma, which was first implicitly mentioned in \cite{Barrett+PRA07}, and stated formally  (but in a form slightly different from the form we present below) in \cite{Bravyi+17,Bravyi+18}. For completeness we include a proof.
\begin{lemma}{(\cite{Barrett+PRA07,Bravyi+17,Bravyi+18})}\label{lemma:affine}
Consider any affine function $q\colon \{0,1\}^3\to  \{0,1\}$ and any three affine functions $q_1\colon \{0,1\}^2\to \{0,1\}$, $q_2\colon \{0,1\}^2\to \{0,1\}$, $q_3\colon \{0,1\}^2\to \{0,1\}$ such that 
\begin{equation}
q_{1}(b_2,b_3)\oplus q_{2}(b_1,b_3)\oplus q_{3}(b_1,b_2)=0
\label{affine-eq1}
\end{equation}
holds for any $(b_1,b_2,b_3)\in \{0,1\}^3$. Then at least one of the four following equalities does not hold:
\begin{eqnarray}
q(0,0,0) &= 0,\label{affine-eq2}\\
q(0,1,1) \oplus q_1(1,1) &= 1,\label{affine-eq3}\\
q(1,0,1) \oplus q_2(1,1) &= 1,\label{affine-eq4}\\
q(1,1,0) \oplus q_3(1,1) &= 1.\label{affine-eq5}
\end{eqnarray}
\end{lemma}
\begin{proof}
Consider any affine function $q\colon\{0,1\}^3\to \{0,1\}$ and any three affine functions $q_1, q_2, q_3\colon\{0,1\}^2\to \{0,1\}$ satisfying Condition
 (\ref{affine-eq1}) for all $(b_1,b_2,b_3)\in \{0,1\}^3$. These four functions can be written as 
\begin{eqnarray}
q(b_1,b_2,b_3) &=&  \alpha_0\oplus \alpha_1 b_1 \oplus \alpha_2 b_2 \oplus \alpha_3 b_3,\label{affine-eqq2}\\
q_1(b_2,b_3)   &=& \beta_0 \oplus \beta_2 b_2 \oplus \beta_3 b_3,\label{affine-eqq3}\\
q_2(b_1,b_3)  &=& \gamma_0 \oplus \gamma_1 b_1 \oplus  \beta_3 b_3,\label{affine-eqq4}\\
q_3(b_1,b_2)  &=& (\beta_0 \oplus \gamma_0) \oplus \gamma_1 b_1 \oplus \beta_2 b_2.\label{affine-eqq5}
\end{eqnarray}
for some coefficients $\alpha_{0}, \alpha_1,\alpha_2,\alpha_3,\beta_0,\beta_2,\beta_3,\gamma_0,\gamma_1\in\{0,1\}$. Assume that these functions satisfy all the four equations (\ref{affine-eq2})-(\ref{affine-eq5}). Equation (\ref{affine-eq2}) implies that $\alpha_0=0$. Consider the quantity
\[
\lambda=q(1,1,0) \oplus q_1(1,1)\oplus q(0,1,1) \oplus q_2(1,1) \oplus q(1,0,1) \oplus q_3(1,1). 
\]
Computing this quantity using the four equations (\ref{affine-eqq2})-(\ref{affine-eqq5}) gives $\lambda=3\alpha_0=0$. On the other hand, computing $\lambda$ using the three equations (\ref{affine-eq3})-(\ref{affine-eq5}) gives $\lambda=1\oplus 1\oplus 1 =1$, which leads to a contradiction and implies that the four equations (\ref{affine-eq2})-(\ref{affine-eq5}) cannot hold simultaneously.
\end{proof}

\subsection{Quantum computation: graph states and their measurements}\label{sec:prelim-q}
{\bf Quantum gates.}
We assume that the reader is familiar with the basics of quantum computation and refer to~\cite{Nielsen+00} for a standard reference. We will use the Hadamard gate $H$ and the Pauli $X$, $Y$ and $Z$ gates:
\[
H=
\frac{1}{\sqrt{2}}\left(
\begin{array}{cc}
1&1\\
1&-1
\end{array}
\right),
\hspace{2mm}
\hspace{2mm}
X=\left(
\begin{array}{cc}
0&1\\
1&0
\end{array}
\right),
\hspace{2mm}
Y=\left(
\begin{array}{cc}
0&-i\\
i&0
\end{array}
\right),
\hspace{2mm}
Z=\left(
\begin{array}{cc}
1&0\\
0&-1
\end{array}
\right),
\]
where $i$ denotes the imaginary unit of complex numbers. Note that $XZ=-ZX=-iY$. We will use two kinds of measurements: measurements in the $X$-basis and measurements in the $Y$-basis, which correspond to projective measurements with observables $X$ and $Y$, respectively.
Concretely, a measurement in the $X$-basis is realized by applying a Hadamard gate to this qubit and then measuring it in the computational basis $\{\ket{0},\ket{1}\}$. A measurement in the $Y$-basis is realized by applying the gate 
\[
\frac{1}{\sqrt{2}}\left(
\begin{array}{cc}
1&-i\\
1&i
\end{array}
\right)
\]
to this qubit and then measuring it in the computational basis.\footnote{The outcome of a measurement in the $X$-basis or the $Y$-basis is often defined as an element in $\{-1,1\}$, i.e., the outcome corresponds to one of two eigenvalues of the observables $X$ and $Y$. In our description the measurement outcome is a bit (the two bits $0$ and $1$ correspond to the two eigenvalues $1$ and $-1$, respectively), which will be more convenient to describe our results.} 
\vspace{2mm}

\noindent
{\bf Graph states.}
Graph states are quantum states that can be described using graphs \cite{Hein+PRA04}. Let $G=(V,E)$ be any undirected graph. The graph state associated with $G$ is the quantum state on $|V|$ qubits obtained by first constructing the state
\[
\bigotimes_{u\in V}\ket{0}_{\mathsf{Q}_u},
\]
where each $\mathsf{Q}_u$ represents a 1-qubit register, 
then applying a Hadamard gate on each register and, finally, applying a Controlled-Z gate on $(\mathsf{Q}_u,\mathsf{Q}_v)$ for any pair $\{u,v\}\in E$. We will write $\ket{G}$ the graph state associated with $G$.

Graph states can equivalently be defined using the stabilizer formalism. For each vertex $u\in V$ define the operator
\[
\pi_u = X_u\otimes \bigotimes_{v\in \Nn(u)}Z_v,
\]
where we use $X_u$ to denote the Pauli operator $X$ applied to Register $\mathsf{Q}_u$ and use $Z_v$ to denote the Pauli operator $Z$ applied to Register $\mathsf{Q}_v$.
Observe that all these operators commute, and 
\[
\pi_u \ket{G}=\ket{G}
\]
for each $u\in G$.
The graph state $\ket{G}$ is thus the simultaneous eigenstate,  associated with the eigenvalue~1, of all these operators . 

\vspace{2mm}

\noindent
{\bf Measurements of graph states.}
The description of graph states using the stabilizer formalism is especially convenient to derive the properties of measurements we describe below (we refer to \cite{Nielsen+00} for details of the general discussion of measurements of stabilizer states and state below only the properties we will use in this paper).

Consider the graph state $\ket{G}$ of a graph $G=(V,E)$. Let $U_X,U_Y\subseteq V$ be any two disjoint subsets of vertices. Assume that we measure Register $\mathsf{Q}_u$, for each vertex $u\in U_X$, in the $X$-basis and measure Register $\mathsf{Q}_v$, for each vertex $v\in U_Y$, in the $Y$-basis. The observable corresponding to this measurement is 
\[
M =\prod_{u\in U_X} X_u\prod_{v\in U_Y} Y_v.
\]
For each $u\in U_X\cup U_Y$, let $z_u\in\{0,1\}$ denote the random variable corresponding to the measurement outcome of the measurement performed on Register $\mathsf{Q}_u$.
Let us denote
\[
z=\bigoplus_{u\in U_X\cup U_Y} z_u
\]
the random variable corresponding to the parity of all the measurement outcomes.
Using the stabilizer formalism it is easy to show that the value of this random variable is as follows:
\begin{itemize}
\item
if $M$ can be written as $M=\prod_{u\in S}\pi_u$ for some set $S\subseteq V$ then $z=0$ with probability 1;
\item 
if $M$ can be written as $M=-\prod_{u\in S}\pi_u$ for some set $S\subseteq V$ then $z=1$ with probability 1;
\item 
if $M$ cannot be written as $M=\prod_{u\in S}\pi_u$ or $M=-\prod_{u\in S}\pi_u$ for some set $S\subseteq V$ then $z=0$ with probability $1/2$ and $z=1$ with probability $1/2$. 
\end{itemize}

\section{Extended Graphs and their Graph States}\label{sec:edge-graph}
In this section we describe the general construction on which our results are based.

For any undirected graph $G=(V,E)$, let $\overline{G}$ denote the graph with $|V|+|E|$ vertices and $2|E|$ edges obtained from $G$ by inserting a vertex at the middle of each edge of $G$. We call $\overline{G}$ the \emph{extended graph of}~$G$. We will write $V^\ast$ the set of inserted vertices and consider $\overline{G}$ as a graph over the vertex set $V\cup V^\ast$.
We refer to Figure \ref{fig1} for an illustration. 

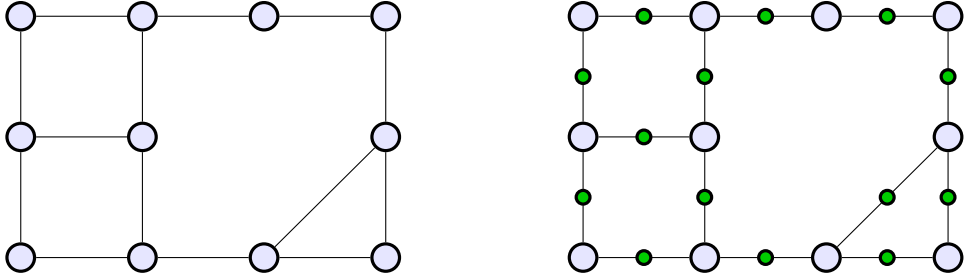
\begin{figure}[th]
\vspace{3mm}
\centering
\begin{tikzpicture}[scale=0.4,rectnode/.style={shape=circle,draw=black,minimum height=20mm},roundnode/.style={circle, draw=black!40, fill=blue!10, very thick, minimum size=2mm}]
\newcommand\XA{4}
\newcommand\YA{4}
    
    \node[roundnode,draw=black,minimum height=3mm] (a0) at (0*\XA,0*\YA) {};    
    \node[roundnode,draw=black,minimum height=3mm] (a1) at (1*\XA,0*\YA) {}; 
    \node[roundnode,draw=black,minimum height=3mm] (a2) at (2*\XA,0*\YA) {}; 
    \node[roundnode,draw=black,minimum height=3mm] (a3) at (3*\XA,0*\YA) {}; 
    \node[roundnode,draw=black,minimum height=3mm] (b0) at (0*\XA,1*\YA) {};    
    \node[roundnode,draw=black,minimum height=3mm] (b1) at (1*\XA,1*\YA) {}; 
    \node[roundnode,draw=black,minimum height=3mm] (b3) at (3*\XA,1*\YA) {}; 
    \node[roundnode,draw=black,minimum height=3mm] (c0) at (0*\XA,2*\YA) {};    
    \node[roundnode,draw=black,minimum height=3mm] (c1) at (1*\XA,2*\YA) {}; 
    \node[roundnode,draw=black,minimum height=3mm] (c2) at (2*\XA,2*\YA) {}; 
    \node[roundnode,draw=black,minimum height=3mm] (c3) at (3*\XA,2*\YA) {}; 
    \path[] 
     (a0) edge (a1)
     (a1) edge (a2)
     (a2) edge (a3)
     (b0) edge (b1)
     (c0) edge (c1)
     (c1) edge (c2)
     (c2) edge (c3)
     (a0) edge (b0)
     (b0) edge (c0)
     (a1) edge (b1)
     (b1) edge (c1)
     (a3) edge (b3)
     (b3) edge (c3)
     (a2) edge (b3);
\end{tikzpicture}
\hspace{20mm}
\begin{tikzpicture}[scale=0.4,rectnode/.style={shape=circle,draw=black,minimum height=20mm},roundnode/.style={circle, draw=black!40, fill=blue!10, very thick, minimum size=2mm},roundnode3/.style={circle, draw=black!100, fill=black!20!green, very thick, minimum size=1mm,scale=0.5}]
\newcommand\XA{4}
\newcommand\YA{4}
    
    \node[roundnode,draw=black,minimum height=3mm] (a0) at (0*\XA,0*\YA) {};    
    \node[roundnode,draw=black,minimum height=3mm] (a1) at (1*\XA,0*\YA) {}; 
    \node[roundnode,draw=black,minimum height=3mm] (a2) at (2*\XA,0*\YA) {}; 
    \node[roundnode,draw=black,minimum height=3mm] (a3) at (3*\XA,0*\YA) {}; 
    \node[roundnode,draw=black,minimum height=3mm] (b0) at (0*\XA,1*\YA) {};    
    \node[roundnode,draw=black,minimum height=3mm] (b1) at (1*\XA,1*\YA) {}; 
    \node[roundnode,draw=black,minimum height=3mm] (b3) at (3*\XA,1*\YA) {}; 
    \node[roundnode,draw=black,minimum height=3mm] (c0) at (0*\XA,2*\YA) {};    
    \node[roundnode,draw=black,minimum height=3mm] (c1) at (1*\XA,2*\YA) {}; 
    \node[roundnode,draw=black,minimum height=3mm] (c2) at (2*\XA,2*\YA) {}; 
    \node[roundnode,draw=black,minimum height=3mm] (c3) at (3*\XA,2*\YA) {}; 
    \node[roundnode3,draw=black,minimum height=3mm] (A0) at (0.5*\XA,0*\YA) {};   
    \node[roundnode3,draw=black,minimum height=3mm] (A1) at (1.5*\XA,0*\YA) {}; 
    \node[roundnode3,draw=black,minimum height=3mm] (A2) at (2.5*\XA,0*\YA) {}; 
    \node[roundnode3,draw=black,minimum height=3mm] (B0) at (0.5*\XA,1*\YA) {};   
    \node[roundnode3,draw=black,minimum height=3mm] (B1) at (2.5*\XA,0.5*\YA) {}; 
    \node[roundnode3,draw=black,minimum height=3mm] (C0) at (0.5*\XA,2*\YA) {};   
    \node[roundnode3,draw=black,minimum height=3mm] (C1) at (1.5*\XA,2*\YA) {}; 
    \node[roundnode3,draw=black,minimum height=3mm] (C2) at (2.5*\XA,2*\YA) {}; 
    \node[roundnode3,draw=black,minimum height=3mm] (D0) at (0*\XA,0.5*\YA) {}; 
    \node[roundnode3,draw=black,minimum height=3mm] (E0) at (0*\XA,1.5*\YA) {}; 
    \node[roundnode3,draw=black,minimum height=3mm] (F0) at (1*\XA,0.5*\YA) {}; 
    \node[roundnode3,draw=black,minimum height=3mm] (G0) at (1*\XA,1.5*\YA) {}; 
    \node[roundnode3,draw=black,minimum height=3mm] (J0) at (3*\XA,0.5*\YA) {}; 
    \node[roundnode3,draw=black,minimum height=3mm] (K0) at (3*\XA,1.5*\YA) {}; 
    \path[] 
     (a0) edge (A0)
     (A0) edge (a1)
     (a1) edge (A1)
     (A1) edge (a2)
     (a2) edge (A2)
     (A2) edge (a3)
     (b0) edge (B0)
     (B0) edge (b1)
     (c0) edge (C0)
     (C0) edge (c1)
     (c1) edge (C1)
     (C1) edge (c2)
     (c2) edge (C2)
     (C2) edge (c3)
     (a0) edge (D0)
     (D0) edge (b0)
     (b0) edge (E0)
     (E0) edge (c0)
	 (a1) edge (F0)
     (F0) edge (b1)
     (b1) edge (G0)
     (G0) edge (c1)
     (B1) edge (b3)
     (B1) edge (a2)
     (a3) edge (J0)
     (J0) edge (b3)
     (b3) edge (K0)
     (K0) edge (c3);
\end{tikzpicture}
\caption{Example for our construction. The graph $G=(V,E)$ is represented on the left. The extended graph $\overline{G}$ is represented on the right. In this figure the large circles represent the vertices in $V$, while the small circles represent the vertices in $V^\ast$.}\label{fig1}
\end{figure}

We now define the concept of \emph{$f$-covering} of a graph. 
\begin{definition}
Let $G=(V,E)$ be an undirected graph and $f\colon V\to \{0,1\}$ be any function such that $|f|$ is even. An $f$-covering of $\overline{G}$ is a set of $|f|/2$ paths 
of $\overline{G}$ such that each vertex in $\{v\in V\:|\: f(v)=1\}$ appears once as an endpoint of one of these paths.
\end{definition}
We refer to Figure \ref{fig1b} for an illustration. 
Note that the $|f|/2$ paths of an $f$-covering do not need to be edge-disjoint.

\begin{figure}[th]
\vspace{3mm}
\centering
\begin{tikzpicture}[scale=0.3,rectnode/.style={shape=circle,draw=black,minimum height=20mm},roundnode/.style={circle, draw=black!40, fill=blue!10, very thick, minimum size=2mm,scale=0.7},roundnode3/.style={circle, draw=black!100, fill=black!20!green, very thick, minimum size=1mm,scale=0.5}]
\newcommand\XA{6}
\newcommand\YA{4}
    
    \node[roundnode,draw=black,minimum height=3mm] (a0) at (0*\XA,0*\YA) {$u_8$};    
    \node[roundnode,draw=black,minimum height=3mm] (a1) at (1*\XA,0*\YA) {$u_9$}; 
    \node[roundnode,draw=black,minimum height=3mm] (a2) at (2*\XA,0*\YA) {\small $u_{10}$}; 
    \node[roundnode,draw=black,minimum height=3mm] (a3) at (3*\XA,0*\YA) {\small $u_{11}$}; 
    \node[roundnode,draw=black,minimum height=3mm] (b0) at (0*\XA,1*\YA) {$u_5$};    
    \node[roundnode,draw=black,minimum height=3mm] (b1) at (1*\XA,1*\YA) {$u_6$}; 
    \node[roundnode,draw=black,minimum height=3mm] (b3) at (3*\XA,1*\YA) {$u_7$}; 
    \node[roundnode,draw=black,minimum height=3mm] (c0) at (0*\XA,2*\YA) {$u_1$};    
    \node[roundnode,draw=black,minimum height=3mm] (c1) at (1*\XA,2*\YA) {$u_2$}; 
    \node[roundnode,draw=black,minimum height=3mm] (c2) at (2*\XA,2*\YA) {$u_{3}$}; 
    \node[roundnode,draw=black,minimum height=3mm] (c3) at (3*\XA,2*\YA) {$u_{4}$}; 
    \node[roundnode3,draw=black,minimum height=3mm] (A0) at (0.5*\XA,0*\YA) {};   
    \node[roundnode3,draw=black,minimum height=3mm] (A1) at (1.5*\XA,0*\YA) {}; 
    \node[roundnode3,draw=black,minimum height=3mm] (A2) at (2.5*\XA,0*\YA) {}; 
    \node[roundnode3,draw=black,minimum height=3mm] (B0) at (0.5*\XA,1*\YA) {};   
    \node[roundnode3,draw=black,minimum height=3mm] (B1) at (2.5*\XA,0.5*\YA) {}; 
    \node[roundnode3,draw=black,minimum height=3mm] (C0) at (0.5*\XA,2*\YA) {};   
    \node[roundnode3,draw=black,minimum height=3mm] (C1) at (1.5*\XA,2*\YA) {}; 
    \node[roundnode3,draw=black,minimum height=3mm] (C2) at (2.5*\XA,2*\YA) {}; 
    \node[roundnode3,draw=black,minimum height=3mm] (D0) at (0*\XA,0.5*\YA) {}; 
    \node[roundnode3,draw=black,minimum height=3mm] (E0) at (0*\XA,1.5*\YA) {}; 
    \node[roundnode3,draw=black,minimum height=3mm] (F0) at (1*\XA,0.5*\YA) {}; 
    \node[roundnode3,draw=black,minimum height=3mm] (G0) at (1*\XA,1.5*\YA) {}; 
    \node[roundnode3,draw=black,minimum height=3mm] (J0) at (3*\XA,0.5*\YA) {}; 
    \node[roundnode3,draw=black,minimum height=3mm] (K0) at (3*\XA,1.5*\YA) {}; 
    \path[] 
     (a0) edge (A0)
     (A0) edge (a1)
     (a1) edge (A1)
     (A1) edge (a2)
     (a2) edge (A2)
     (A2) edge (a3)
     (b0) edge (B0)
     (B0) edge (b1)
     (c0) edge (C0)
     (C0) edge (c1)
     (c1) edge (C1)
     (C1) edge (c2)
     (c2) edge (C2)
     (C2) edge (c3)
     (a0) edge (D0)
     (D0) edge (b0)
     (b0) edge (E0)
     (E0) edge (c0)
	 (a1) edge (F0)
     (F0) edge (b1)
     (b1) edge (G0)
     (G0) edge (c1)
     (B1) edge (b3)
     (B1) edge (a2)
     (a3) edge (J0)
     (J0) edge (b3)
     (b3) edge (K0)
     (K0) edge (c3);
     \draw[very thick,red] (b0) edge (E0);
     \draw[very thick,red] (E0) edge (c0);
     \draw[very thick,red] (c0) edge (C0);
     \draw[very thick,red] (C0) edge (c1);
     \draw[very thick,red] (c3) edge (K0);
     \draw[very thick,red] (K0) edge (b3);
     \draw[very thick,red] (b3) edge (B1);
     \draw[very thick,red] (B1) edge (a2);
\end{tikzpicture}
\caption{Illustration of the concept of $f$-covering. Here $V=\{u_1,\ldots,u_{11}\}$ and $f\colon V\to\{0,1\}$ is defined as follows:  $f(u_2)=f(u_4)=f(u_5)=f(u_{10})=1$ and $f(u_1)=f(u_3)=f(u_6)=f(u_7)=f(u_8)=f(u_9)=f(u_{11})=0$. The two paths depicted in red form an $f$-covering.}\label{fig1b}
\end{figure}
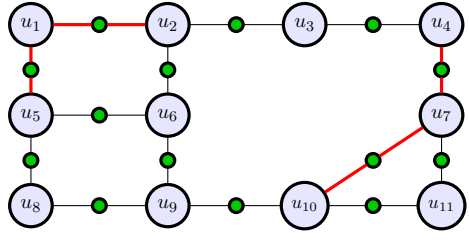

\begin{figure}[th]
\begin{center}
\fbox{
\begin{minipage}{13 cm} 
\begin{itemize}
\item[1.]
Construct the graph state over $\overline{G}$.
\item[2.]
For each $v\in V$ such that $f(v)=1$, measure the qubit of the node $v$ in the $Y$-basis.
Measure the qubits of all the other nodes of $\overline{G}$ in the $X$-basis.
\end{itemize}
\end{minipage}
}
\end{center}\vspace{-4mm}
\caption{The process $\Pp(G,f)$.}\label{fig:process}
\end{figure}

Given a graph $G=(V,E)$ and a function $f\colon V\to\{0,1\}$, consider the process $\Pp(G,f)$ described in Figure~\ref{fig:process}. For any vertex $v\in V\cup V^\ast$, let $z_v$ denote the random variable corresponding to the outcome of the measurement performed on the qubit of node $v$.
The following two theorems describe the correlations among these random variables.
\begin{theorem}\label{th:cond1}
For any cycle $\Cc$ of $\overline{G}$ the following equality holds with probability 1:
\begin{equation}\label{eq:cond-cycle}
\bigoplus_{v\in \Cc\cap V^\ast} z_{v} = 
0.
\end{equation}
\end{theorem}
\begin{proof}
In Process $\Pp(G,f)$ all the vertices in $\Cc\cap V^\ast$ are measured in the $X$-basis. Since $\Cc$ is a cycle we have
\[
\prod_{v\in \Cc\cap V^\ast} \pi_{v}=
\prod_{v\in \Cc\cap V^\ast} X_{v},
\]
which is the measurement operator corresponding to this measurement.
The discussion of Section \ref{sec:prelim-q} implies that the parity of all the measurement outcomes is always zero, as claimed.
\end{proof}

\begin{theorem}\label{th:cond2}
Assume that $|f|$ is even and let $\{p_1,\ldots,p_{|f|/2}\}$ be any $f$-covering of $\overline{G}$. Let us write
\[
z_V = \bigoplus_{v\in V} z_v.
\]
Then the following equality holds with probability 1:
\begin{equation}\label{eq:cond}
z_{V}\oplus \bigoplus_{i=1}^{|f|/2}\bigoplus_{v\in p_i\cap V^\ast} z_{v} = 
\left\{
\begin{tabular}{cl}
$0$&\textrm{ if }$|f|\bmod 4=0$,\\
$1$&\textrm{ if }$|f|\bmod 4=2$.
\end{tabular}
\right.
\end{equation}
\end{theorem}
\begin{proof}
Let  $V_1=\{u_1,\ldots,u_{|f|/2},v_1,\ldots,v_{|f|/2}\}\subseteq V$ denote the set of vertices that appear as an endpoint of one of the paths. 
Let $V_2\subseteq V^\ast$ denote the set of vertices in $V^\ast$ that appear 
on an odd number of paths (remember that the paths in an $f$-covering do not need to be disjoint). Note that the equation we want to show (Equation (\ref{eq:cond})) can be rewritten as 
\begin{equation}\label{eq:cond2}
z_{V}\oplus \bigoplus_{v\in V_2} z_{v} = 
\left\{
\begin{tabular}{cl}
$0$&\textrm{ if }$|f|\bmod 4=0$,\\
$1$&\textrm{ if }$|f|\bmod 4=2$.
\end{tabular}
\right.
\end{equation}

From the definition of an $f$-covering, we have $V_1=\{v\in V\:|\: f(v)=1\}$, and thus in Process $\Pp(G,f)$ all the vertices in $V_1$ are measured in the $Y$-basis, while the vertices in $V\setminus V_1$ and the vertices in $V_2$ are measured in the $X$-basis. The observable corresponding to this measurement is thus
\begin{equation}\label{eq:stab}
\prod_{u\in V_1}Y_u
\prod_{v\in (V\setminus V_1)\cup V_2}X_v.
\end{equation}

Observe that
\[
\prod_{v\in V}\pi_v=
\prod_{v\in V}X_v.
\]
This simple but crucial property follows from our construction: $\overline{G}$ is obtained from $G$ by inserting a vertex on each edge of $G$.
For each $i\in\{1,\ldots,|f|/2\}$ we also have 
\[
\prod_{v\in p_i\cap V^\ast} \pi_v =
Z_{u_i}\left(\prod_{v\in p_i\cap V^\ast}X_v \right)
Z_{v_i}.
\]
Thus 
\begin{eqnarray*}
\left(\prod_{v\in V}\pi_v\right)\times
\left(\prod_{i=1}^{|f|/2}\prod_{v\in p_i\cap V^\ast} \pi_v\right)&=&
\left(\prod_{u\in V_1} X_{u}Z_{u}\right)\times 
\left(\prod_{v\in V\setminus V_1}X_{v}\right)\times
\left(\prod_{v\in V_2} X_{v}\right)\\
&=&(-1)^{|f|/2}\prod_{u\in V_1}Y_u
\prod_{v\in (V\setminus V_1)\cup V_2}X_v.
\end{eqnarray*}

When $|f|\bmod 4 =0$ the observable of Equation (\ref{eq:stab}) can then be written as a product of generators of the graph state, and thus the parity of all the measurement outcomes is $0$.  When $|f|\bmod 4 =2$ the additive inverse of this observable can be written as a product of generators of the graph states, and thus the parity of all the measurement outcomes is $1$. This proves Equation (\ref{eq:cond2}), and thus Equation (\ref{eq:cond}).
\end{proof}

\paragraph{Remark.}
The conditions of Equations (\ref{eq:cond-cycle}) for all the cycles $\Cc$ of $\overline{G}$ and the condition of Equation~(\ref{eq:cond}) together actually completely characterize the distribution of the outcomes of $\Pp(G,f)$: the variables $\{z_v\}_{v\in V\cup V^\ast}$ are uniformly distributed over the set of all values satisfying all these equations. Note that when $G$ is a connected graph then this corresponds to satisfying exactly $|E|-|V|+2$ independent linear equations. Indeed, $|E|-(|V|-1)$ equations suffice to guarantee that Equation (\ref{eq:cond-cycle}) holds for all the cycles $\Cc$ of $\overline{G}$, as can be seen by considering a spanning tree of $G$: the spanning tree contains $|V|-1$ edges and each of the remaining $|E|-(|V|-1)$ edges gives rise to a cycle in $G$ (and thus to a new linear equation) when added to the spanning tree. A similar characterization can be easily obtained when $G$ is not connected as well, by considering separately each connected component.

\section{Description of the Relation \texorpdfstring{$\boldsymbol{R}$}{R}}\label{sec:construction}
In this section we describe the computational problem we use to prove Theorem \ref{th:main}.

\subsection{Our graph construction}\label{subsec:graph}
For any even positive integer $d$, we explain how to construct two graphs $\Gg_{d}$ and $\overline{\Gg}_{d}$ that we will use to define the computational problems. The construction is illustrated in Figure \ref{fig2}.

\begin{figure}[!t]
\vspace{3mm}
\centering
\begin{tikzpicture}[scale=0.45,roundnode/.style={circle, draw=black!100, fill=blue!10, very thick, minimum size=1mm,scale=0.5},roundnode2/.style={circle, draw=black!100, fill=blue!90, very thick, minimum size=1mm,scale=0.5},roundnode3/.style={circle, draw=black!100, fill=black!20!green, very thick, minimum size=1mm,scale=0.3}]
\newcommand\XA{2}
\newcommand\YA{2}

\foreach \i in {-1,3,9}{
\foreach \j in {0,4,-6}{
\draw[dashed,draw=black,fill=none,thick] (\i*\XA-0.2*\XA,\j*\YA-0.2*\YA) rectangle (\i*\XA+3.2*\XA,\j*\YA+3.2*\YA) ;
}}

\foreach \i in {-2,-1,0,2,3,4,8,9,10}{
\foreach \j in {-1,0,1,3,4,5,-5,-6,-7}{
\draw [very thin,black] (\i*\XA+1.5*\XA,\j*\YA+1.5*\XA) -- (\i*\XA+1*\XA,\j*\YA+1*\XA){};
\draw [very thin,black] (\i*\XA+1.5*\XA,\j*\YA+1.5*\XA) -- (\i*\XA+2*\XA,\j*\YA+1*\XA){};
\draw [very thin,black] (\i*\XA+1.5*\XA,\j*\YA+1.5*\XA) -- (\i*\XA+1*\XA,\j*\YA+2*\XA){};
\draw [very thin,black] (\i*\XA+1.5*\XA,\j*\YA+1.5*\XA) -- (\i*\XA+2*\XA,\j*\YA+2*\XA){};
\node[roundnode2] (a0) at (\i*\XA+1.5*\XA,\j*\YA+1.5*\XA) {}; 
\node[roundnode3] (a0) at (\i*\XA+1.75*\XA,\j*\YA+1.75*\XA) {}; 
\node[roundnode3] (a0) at (\i*\XA+1.75*\XA,\j*\YA+1.25*\XA) {}; 
\node[roundnode3] (a0) at (\i*\XA+1.25*\XA,\j*\YA+1.25*\XA) {}; 
\node[roundnode3] (a0) at (\i*\XA+1.25*\XA,\j*\YA+1.75*\XA) {}; 
}}

\draw [very thick,black,<-] (-.7*\XA,6.5*\YA) .. controls (-1.4*\XA,7.0*\YA) .. (-1.4*\XA,7.5*\YA){};
\node[draw=none,fill=none,] at (-1.5*\XA,7.7*\XA) {$u_{11}$};

\draw [very thick,black,<-] (11.7*\XA,6.5*\YA) .. controls (12.5*\XA,7.0*\YA) .. (12.4*\XA,7.5*\YA){};
\node[draw=none,fill=none,] at (12.5*\XA,7.7*\XA) {$u_{1k}$};

\draw [very thick,black,<-] (-.7*\XA,-5.5*\YA) .. controls (-1.4*\XA,-6.0*\YA) .. (-1.4*\XA,-6.5*\YA){};
\node[draw=none,fill=none,] at (-1.5*\XA,-6.7*\XA) {$u_{k1}$};

\draw [very thick,black,<-] (11.7*\XA,-5.5*\YA) .. controls (12.5*\XA,-6.0*\YA) .. (12.4*\XA,-6.5*\YA){};
\node[draw=none,fill=none,] at (12.5*\XA,-6.7*\XA) {$u_{kk}$};

\foreach \i in {-6,-5,-4,-3,0,1,2,3,4,5,6,7} {
        \draw [very thin,black] (-1*\XA,\i*\YA) -- (6.5*\XA,\i*\YA){};
}
\foreach \i in {-6,-5,-4,-3,0,1,2,3,4,5,6,7} {
        \draw [very thin,black] (8.5*\XA,\i*\YA) -- (12*\XA,\i*\YA){};
}     
     
\foreach \i in {-1,...,6,9,10,11,12} {     
        \draw [very thin,black] (\i*\XA,-0.5*\YA) -- (\i*\XA,7*\YA){};
}

\foreach \i in {-1,...,6,9,10,11,12} {     
        \draw [very thin,black] (\i*\XA,-6*\YA) -- (\i*\XA,-2.5*\YA){};
}

\foreach \i in {-1,...,6,9,10,11,12} {     
        \draw [dotted,very thick,black] (\i*\XA,-2*\YA) -- (\i*\XA,-1*\YA){};
}

\foreach \i in {-6,-5,-4,-3,0,1,2,3,4,5,6,7} {
        \draw [dotted,very thick,black](7*\XA,\i*\YA) -- (8*\XA,\i*\YA){};
}    

\draw [very thick,blue] (6.5*\XA,-6*\YA) -- (-1*\XA,-6*\YA) -- (-1*\XA,-2.5*\YA){}; 
\draw [very thick,blue] (-1*\XA,-0.5*\YA) -- (-1*\XA,7*\YA) -- (6.5*\XA,7*\YA){}; 
\draw [very thick,blue] (8.5*\XA,7*\YA) -- (12*\XA,7*\YA) -- (12*\XA,-0.5*\YA){}; 
\draw [very thick,blue] (12*\XA,-2.5*\YA) -- (12*\XA,-6*\YA) -- (8.5*\XA,-6*\YA){}; 

\foreach \i in {-1,0,1,2,3,4,5,6,9,10,11,12} {
	\foreach \j in {-6,-5,-4,-3,0,1,2,3,4,5,6,7} {
			\node[roundnode] (a0) at (\i*\XA,\j*\YA) {};   
	}
}

\foreach \i in {-0.5,0.5,1.5,2.5,3.5,4.5,5.5,9.5,10.5,11.5} {
	\foreach \j in {-6,-5,-4,-3,0,1,2,3,4,5,6,7} {
			\node[roundnode3] (a0) at (\i*\XA,\j*\YA) {};   
	}
}

\foreach \i in {-1,0,1,2,3,4,5,6,9,10,11,12} {
	\foreach \j in {-5.5,-4.5,-3.5,0.5,1.5,2.5,3.5,4.5,5.5,6.5} {
			\node[roundnode3] (a0) at (\i*\XA,\j*\YA) {};   
	}
}

\draw [thin,black,<->] (-1.1*\XA,7.5*\YA) -- (2.1*\XA,7.5*\YA){};
\draw [thin,black,<->] (-1.1*\XA,8.3*\YA) -- (12.1*\XA,8.3*\YA){};
\draw [thin,black,<->] (2.9*\XA,7.5*\YA) -- (6.1*\XA,7.5*\YA){};
\draw [thin,black,<->] (8.9*\XA,7.5*\YA) -- (12.1*\XA,7.5*\YA){};
\node[draw=none,fill=none,] at (0.5*\XA,7.9*\YA) {$d$};
\node[draw=none,fill=none,] at (10.5*\XA,7.9*\YA) {$d$};
\node[draw=none,fill=none,] at (4.5*\XA,7.9*\YA) {$d$};
\node[draw=none,fill=none,] at (5.1*\XA,8.6*\YA) {$d^3$};

\draw [thin,black,<->] (-1.6*\XA,3.9*\YA) -- (-1.6*\XA,7.1*\YA){};
\draw [thin,black,<->] (-1.6*\XA,-0.1*\YA) -- (-1.6*\XA,3.1*\YA){};
\draw [thin,black,<->] (-1.6*\XA,-6.1*\YA) -- (-1.6*\XA,-2.9*\YA){};
\draw [thin,black,<->] (-2.4*\XA,-6.1*\YA) -- (-2.4*\XA,7.1*\YA){};
\node[draw=none,fill=none,] at (-2*\XA,5.5*\YA) {$d$};
\node[draw=none,fill=none,] at (-2*\XA,1.5*\YA) {$d$};
\node[draw=none,fill=none,] at (-2*\XA,-4.5*\YA) {$d$};
\node[draw=none,fill=none,] at (-2.9*\XA,0.5*\YA) {$d^3$};
\vspace{-2mm}
\end{tikzpicture}
\vspace{-2mm}
\caption{The graph $\overline{\Gg}_{d}$, here represented for $d=4$. The vertices in $V_d^1$ are represented in white, the vertices in $V_d^2$ are represented in blue and the vertices in $V_d^\ast$ are represented in green. The blue line represents the external border of the graph. The dashed squares represent the boxes.}\label{fig2}
\end{figure}
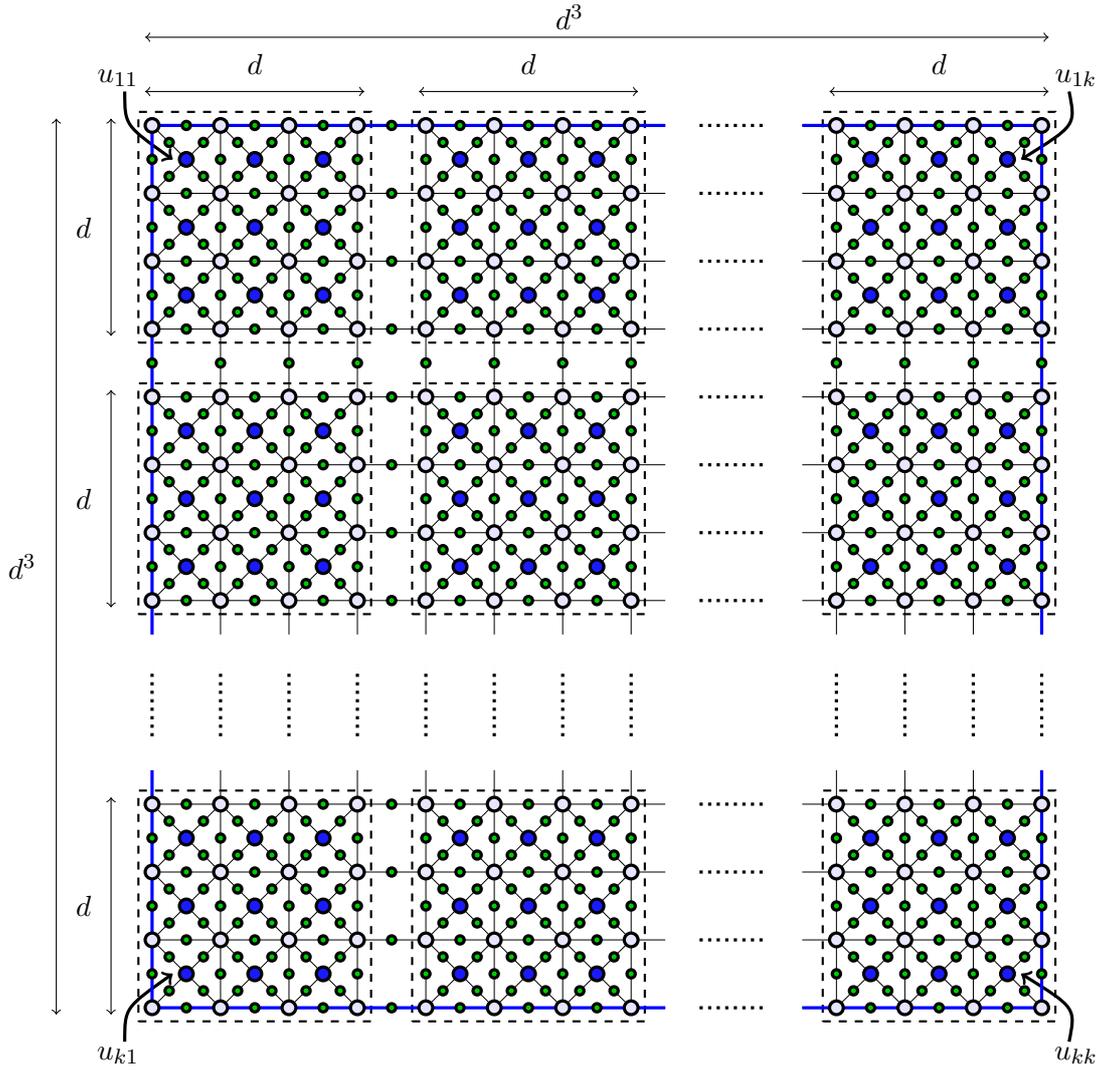

The graph $\Gg_{d}$ is the graph with vertex set $V_d=V_d^1\cup V_d^2$ defined as follows. We start with a $d^3\times d^3$ square grid and denote $V^1_d$ the set of vertices of this grid (observe that $|V_d^1|=d^6$). 
This grid can be divided into $d^4$ contiguous square regions each of size $d\times d$. We call each region a box. In each box we place a vertex at the center of each $1\times 1$ square and connect it to the four corners of the square. Let $V_d^2$ denote the set of all these new vertices. We have $|V_d^2|=d^4(d-1)^2$. It will be convenient to denote those vertices $u_{ij}$ for $i,j\in\{1,\ldots,k\}$, where $k=d^2(d-1)$, with the index $i$ representing the horizontal position and the index $j$ representing the vertical position. This completes the description of $\Gg_{d}$.

The graph $\overline{\Gg}_{d}$ is obtained from $\Gg_{d}$ by the construction described in Section \ref{sec:edge-graph}: one vertex is inserted on each edge of $\Gg_{d}$. Let $V^\ast_d$ denote the introduced vertices. Note that 
$|V^\ast_d|=2d^3(d^3-1)+4d^4(d-1)^2$. Let us denote $\overline{V}_d=V^1_d\cup V^2_d \cup V^\ast_d$ the set of all vertices in  $\overline{\Gg}_d$ and write  
\[
n = |\overline{V}_d|=\Theta(d^6).
\]
For any vertex $u\in V_d$, let $\Bo(u)$ denote the unique $d\times d$ box in which $u$ is included. Finally, we denote $\bord(\overline{\Gg}_d)$ the external border of the graph, i.e., the perimeter of the whole grid.

\subsection{Definition of the relation}\label{sub:def}
Let $d$, $k$ and $n$ be as in Section \ref{subsec:graph}.
Given a matrix $A\in\{0,1\}^{k\times k}$, consider the process $\Pp_d(A)$ described in Figure \ref{fig3}.
\begin{figure}[!ht]
\begin{center}
\fbox{
\begin{minipage}{14.5 cm} 
\begin{itemize}
\item[1.]
Construct the graph state over $\overline{\Gg}_{d}$.
\item[2.]
For each vertex $u_{ij}\in V_d^2$, measure the qubit of the vertex $u_{ij}$ in the $X$-basis if $A_{ij}=0$, and measure it in the $Y$-basis if $A_{ij}=1$. 

For each vertex $u\in V_d^1\cup V_d^\ast$, measure the qubit of the vertex $u$ in the $X$-basis.
\end{itemize}
\end{minipage}
}
\end{center}\vspace{-4mm}
\caption{The process $\Pp_d(A)$.}\label{fig3}
\end{figure}

In this process each node of $\overline{\Gg}_d$ performs a measurement and outputs one bit. We represent the whole output by a binary string of length $n$ by fixing an arbitrary ordering of the $n$ nodes of $\overline{\Gg}_d$. With this representation of measurement outcomes as strings, let
\[
\Lambda_d(A)\subseteq \{0,1\}^n
\] 
denote the set of all the strings that occur with non-zero probability in Process $\Pp_d(A)$.

\paragraph{Definition of the relation $\boldsymbol{R}$.}
For any even positive integer $d$, the computational problem that we consider is as follows: given a matrix $A\in\{0,1\}^{k\times k}$ as input, where $k=d^2(d-1)$, compute a string from $\Lambda_d(A)$. Note that since $|\Lambda_d(A)|> 1$ there are more than one valid output. This computational problem corresponds to the relation
\[
R=\big\{(A,z)\:|\: A\in \{0,1\}^{k\times k} \textrm{ and } z\in \Lambda_d(A)\big\}\subseteq \{0,1\}^{k\times k}\times \{0,1\}^n.
\]
By setting $m=k^2$ and identifying $\{0,1\}^{k\times k}$ with $\{0,1\}^{m}$, we interpret $R$ as a subset of $\{0,1\}^m\times \{0,1\}^n$.
This relation $R$ is the relation 
that appears in the statement of Theorem \ref{th:main}. To avoid confusion it will be preferable to make explicit the dependence on the parameter $d$. We will thus denote this relation by $R_d$ instead of $R$ in the next sections.

\section{Proof of Theorem \ref{th:main}}\label{sec:proof}
In this section we prove Theorem \ref{th:main}. Let $R_d$ be the relation defined in Section \ref{sub:def}.

In the quantum setting, the computational problem corresponding to $R_d$ can obviously be solved by directly implementing the process $\Pp_d(A)$. This can be done by a constant-depth quantum circuit since the graph $\overline{\Gg}_d$, which has constant degree, can be constructed in constant depth. Note that the description of the quantum circuit can be computed easily, e.g., by a logarithmic-space classical Turing machine.

We now show the classical lower bound, i.e., show that any classical circuit of sublogarithmic depth with bounded-fanin gates cannot output a string in $\Lambda_d(A)$ with high probability on a non-negligible fraction of the inputs $A$. For concreteness (and without loss of generality) we will assume in this section that all the gates in the classical circuit have fanin at most 2.

Consider any randomized classical circuit $C_d$, with gates of fanin at most 2, of depth at most $\frac{1}{8} \log_2 m$ for the relation~$R_d$. The circuit has $m=k^2=\Theta(d^6)$ input wires to receive the matrix~$A$ and $n$ output wires. Remember that $n=\Theta(d^6)$. To simplify the presentation we assume that $d$ is large enough so that the inequality 
\begin{equation}\label{ineq}
3n^{1/7}< d-2
\end{equation} 
holds. In Section \ref{subseqproof1} below we show how to associate the wires of $C_d$ to the nodes of $\overline{\Gg}_d$. In Section \ref{subseqproof2} we present technical results that exploit this correspondence. Finally, in Section \ref{subseqproof3} we give an upper bound on the success probability of $C_d$ and conclude the proof of Theorem \ref{th:main}.
\subsection{Correspondence between \texorpdfstring{$\boldsymbol{C_d}$}{the circuit} and \texorpdfstring{$\boldsymbol{\overline{\Gg}_d}$}{the graph}}\label{subseqproof1}
We associate the wires of $C_d$ to the nodes of $\overline{\Gg}_d$ in the following way. For any vertex $u_{ij}\in V^2_d$, we denote $x_{u_{ij}}$ the input wire of $C_d$ that receives the entry $A_{ij}$ of $A$. For any vertex $u\in \overline{V}_d$, we denote $z_u$ the output wire of $C_d$ that should output the outcome of the measurement performed at vertex $u$.

For any vertex $u\in \overline{V}_d$, we denote $L(z_u)$ the set of all vertices $v\in V_d^2$ such that the input wire $x_{v}$ is in the lightcone of $z_u$ (i.e., the value of $z_{u}$ depends on the value of $x_v$). For any $u\in V^2_d$, we denote $L(x_u)$ the set of all vertices $v\in \overline{V}_d$ such that the output wire $z_{v}$ is in the lightcone of $x_u$ (i.e., the value of $z_{v}$ depends on the value of $x_u$). Since the depth of $C_d$ is at most $\frac{1}{8} \log_2 m$ and since each gate of $C_d$ has fanin at most 2, we have $|L(z_u)|\le m^{1/8}\le n^{1/8}$ for each $u\in \overline{V}_d$.
Let us define the set
\[
\Gamma = \{u\in V_d^2 \:|\: L(x_u)> n^{1/7} \}.
\]
Since the number of input wires is $|V_d^2|=\Theta(n)$, a simple counting argument shows that $|\Gamma|=O(n^{55/56})$, i.e., most input wires have small lightcones as well. 

Define the sets $\Uu,\Vv, \Ww\subseteq V_d^2$ as follows:
\begin{align*}
\Uu&= \big\{u_{ij}\:|\:i\in\{1,\ldots,\floor{k/3}\} \textrm { and } j\in\{1,\ldots,\floor{k/3}\}\big\}\setminus\Gamma,\\
\Vv&= \big\{u_{ij}\:|\:i\in\{\ceil{2k/3},\ldots,k\} \textrm { and } j\in\{1,\ldots,\floor{k/3}\}\big\}\setminus\Gamma,\\
\Ww&= \big\{u_{ij}\:|\:i\in\{\ceil{2k/3},\ldots,k\} \textrm { and } j\in\{\ceil{2k/3},\ldots,k\}\big\}\setminus\Gamma.
\end{align*}
These three sets represent the vertices in $V_d^2\cap \Gamma$ that are in the upper left part, the upper right part, and the lower right part of the graph $\overline{\Gg}_d$, respectively. From the above discussion we have $|\Uu|=\Theta(n)$, $|\Vv|=\Theta(n)$ and $|\Ww|=\Theta(n)$. 

\subsection{Graph-theoretic arguments}\label{subseqproof2}
We start with a first lemma, which is similar to \cite[Claim 6]{Bravyi+17}. 
\begin{lemma}\label{lemma1}
The number of triples $(u,v,w)\in \Uu\times\Vv\times\Ww$ such that the three conditions
\begin{itemize}
\item
$L(x_u)\cap \Bo(v)=\emptyset$ and $L(x_u)\cap \Bo(w)=\emptyset$;
\item
$L(x_v)\cap \Bo(u)=\emptyset$ and $L(x_v)\cap \Bo(w)=\emptyset$;
\item
$L(x_w)\cap \Bo(u)=\emptyset$ and $L(x_w)\cap \Bo(v)=\emptyset$.
\end{itemize}
do not simultaneously hold is $O(n^{2+10/21})$.
\end{lemma}
\begin{proof}
Observe that for each $u\in \Uu$, there are at most $n^{1/7}$ boxes that intersect $L(x_u)$. Since each box contains $(d-1)^2=O(n^{1/3})$ vertices in $V_d^2$, there are at most $O(n^{10/21})$ vertices $v\in \Vv$ such that $\Bo(v)$ intersects $L(x_u)$.
Assume that we choose a vertex $v$ uniformly at random in $\Vv$. Then we have 
\[
\Pr_{v\in \Vv}\left[L(x_u)\cap \Bo(v)\neq \emptyset \right]=O\left(n^{-11/21}\right).
\]
Applying the union bound shows that if we choose a triple $(u,v,w)$ uniformly at random in $\Uu\times\Vv\times\Ww$, then the probability that this triple does not satisfy all the three conditions of the lemma is $O(n^{-11/21})$. Since $|\Uu\times\Vv\times\Ww|=\Theta(n^3)$, we thus obtain the statement of the lemma.
\end{proof}

The following simple lemma will be crucial for our analysis.
\begin{lemma}\label{lemma2}
The number of triples $(u,v,w)\in \Uu\times\Vv\times\Ww$  such that the three lightcones $L(x_u)$, $L(x_v)$ and $L(x_w)$ are not pairwise disjoint is $O(n^{2+2/7})$. 
\end{lemma}
\begin{proof}
Let $t\in\overline{V}_d$ be any vertex of $\overline\Gg_d$. When choosing $(u,v)$ uniformly at random in $\Uu\times \Vv$, the probability that $t$ is in $L(x_u)\cap L(x_v)$ is $O((n^{1/7}/n)^2)=O(n^{-12/7})$. By the union bound this implies that when choosing a triple $(u,v,w)$ uniformly at random in $\Uu\times \Vv\times \Ww$, the probability that $x$ is in more than one of the three lightcones $L(x_u)$, $L(x_v)$ and $L(x_w)$ is $O(n^{-12/7})$ as well. By the union bound again, we conclude that when choosing $(u,v,w)$ uniformly at random in $\Uu\times \Vv\times \Ww$, the probability that the three lightcones $L(x_u)$, $L(x_v)$ and $L(x_w)$ are not pairwise disjoint is $O(n^{-5/7})$.   
\end{proof}

The following proposition is the main result of this subsection.
\begin{proposition}\label{prop1}
There exists a triple of vertices $(u,v,w)\in \Uu\times \Vv\times \Ww$ such that all the following conditions hold:
\begin{itemize}
\item[(i)]
the lightcones $L(x_u)$, $L(x_v)$ and $L(x_w)$ are pairwise disjoint;
\item[(ii)]
there exists a cycle $\Cc$ containing $u$, $v$ and $w$ such that 
\begin{itemize}
\item[(ii-a)]
$\Cc$ does not use any edge from the external border $\bord({\overline{\Gg}_d})$;
\item[(ii-b)]
$\Cc\cap V_d^2=\{u,v,w\}$;
\item[(ii-c)]
$q_1\cap L(x_w)=\emptyset$, $q_2\cap L(x_u)=\emptyset$ and $q_3\cap L(x_v)=\emptyset$, where $q_1$ denotes the direct path\footnote{There are two paths from $v$ to $w$ in the cycle $\Cc$: one path going via $u$ and one path not using $u$. The direct path is the latter.} from $v$ to $w$ in the cycle $\Cc$,  $q_2$ denotes the direct path from $u$ to $w$ in $\Cc$ and let $q_3$ denote the direct path from $u$ to $v$ in $\Cc$.
\end{itemize}
\end{itemize}
\end{proposition}
\begin{proof}
Lemmas \ref{lemma1} and \ref{lemma2} imply that among the $\Theta(n^3)$ triples $(u,v,w)\in \Uu\times \Vv\times \Ww$ there exists one triple such that Condition (i) and the three conditions of Lemma \ref{lemma1} simultaneously hold. Let us fix such a triple.

Let $u_1,\ldots,u_{d-2}$ denote the vertices on the right border\footnote{To simplify the presentation we exclude the two corners at the extremities of each border.} of $\Bo(u)$ and $u'_1,\ldots,u'_{d-2}$ denote the vertices on the bottom border of $\Bo(u)$. Similarly, let $v_1,\ldots,v_{d-2}$ denote the vertices on the left border of $\Bo(v)$ and $v'_1,\ldots,v'_{d-2}$ denote the vertices on the bottom border of $\Bo(v)$. Finally, let $w_1,\ldots,w_{d-2}$ denote the vertices on the top border of $\Bo(w)$ and $w'_1,\ldots,w'_{d-2}$ denote the vertices on the left border of $\Bo(w)$. We refer to Figure \ref{fig5} for an illustration.

We can construct a path $p^1_i$ from $u_i$ to $v_i$, a path $p^2_i$ from $v'_i$ to $w_i$ and a path $p^3_i$ from $w'_i$ to $u'_i$, for each $i\in\{1,\ldots,d-2\}$, so that the $3(d-2)$ paths constructed are disjoint, do not use any edge on the border $\bord({\overline{\Gg}_d})$, do not go through any vertex in $V^2_d$, and do not contain any vertex in $\Bo(u)\cup\Bo(v)\cup\Bo(w)$ except their endpoints. From Inequality (\ref{ineq}) and since the three lightcones $L(x_u)$, $L(x_v)$ and $L(x_w)$ do not have size larger than $n^{1/7}$,
there necessarily exist three indices $i_1,i_2,i_3\in\{1,\ldots,d-2\}$ such that the three paths $p^1_{i_1}$, $p^2_{i_2}$ and $p^3_{i_3}$ do not contain any vertex in $L(x_u)\cup L(x_v)\cup L(x_v)$. Finally, observe that these three paths can be completed (avoiding all vertices in $V_d^2\setminus\{u,v,w\}$) to obtain a cycle 
\[
u\xrightarrow{}u_{i_1}
\xrightarrow{p^1_{i_1}} v_{i_1}
\xrightarrow{} v
\xrightarrow{} v_{i_2}'
\xrightarrow{p^2_{i_2}} w_{i_2}
\xrightarrow{} w
\xrightarrow{} w'_{i_3}
\xrightarrow{p^3_{i_3}} u'_{i_3}
\xrightarrow{} u
\]
that satisfies Conditions (ii-a), (ii-b) and (ii-c). 
See Figure \ref{fig5} for an illustration.
Note that Condition~(ii-c) can be guaranteed due to the fact that $(u,v,w)$ satisfies the three conditions from Lemma \ref{lemma1}.
\end{proof}

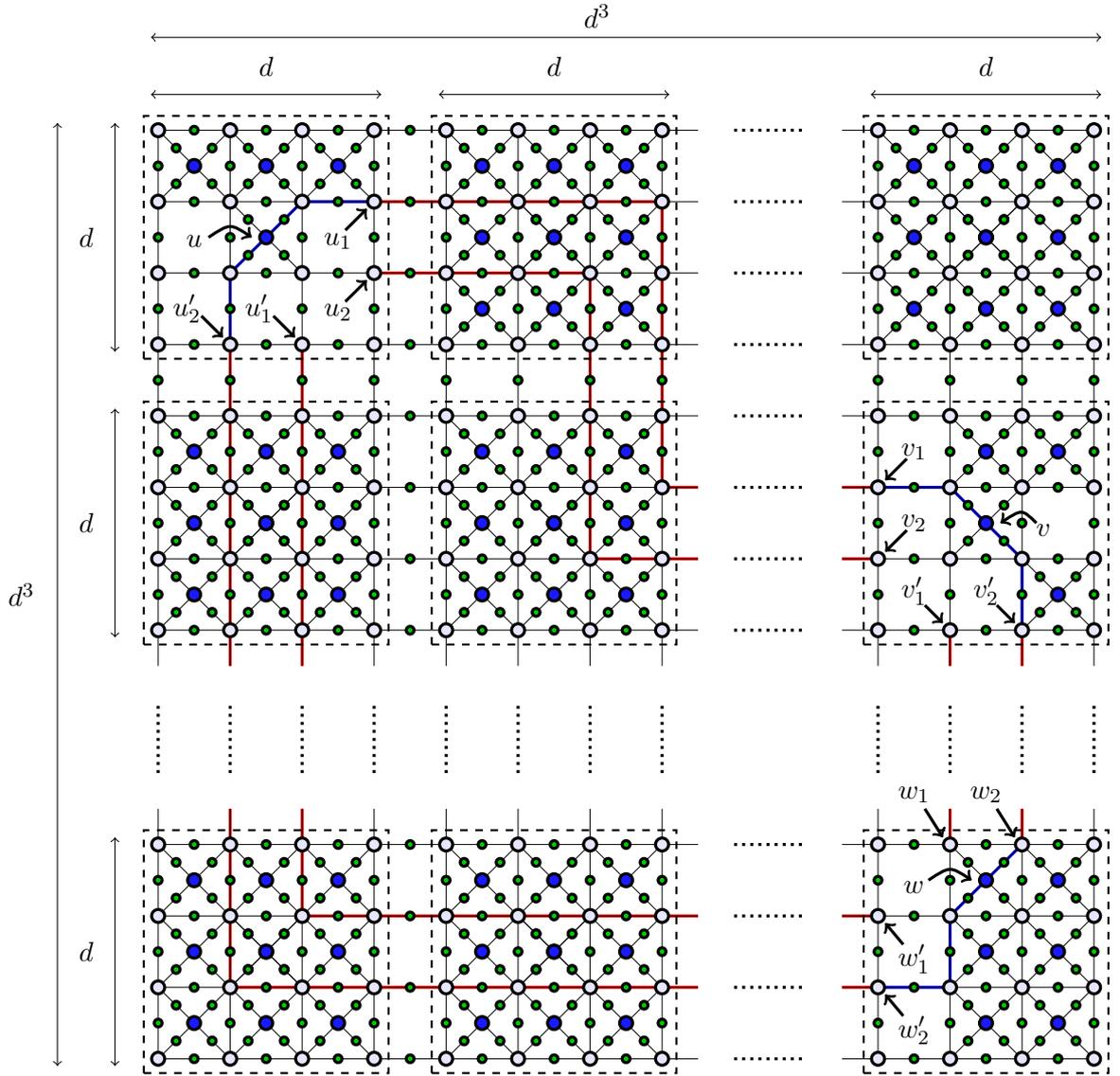
\begin{figure}[!t]
\vspace{3mm}
\centering
\begin{tikzpicture}[scale=0.5,roundnode/.style={circle, draw=black!100, fill=blue!10, very thick, minimum size=1mm,scale=0.5},roundnode2/.style={circle, draw=black!100, fill=blue!90, very thick, minimum size=1mm,scale=0.5},roundnode3/.style={circle, draw=black!100, fill=black!20!green, very thick, minimum size=1mm,scale=0.3}]
\newcommand\XA{2}
\newcommand\YA{2}

\draw [very thick,blue] (0.5*\XA,5.5*\YA)  --  (1*\XA,6*\YA) -- (2*\XA,6*\YA) {};

\draw [very thick,red] (2*\XA,6*\YA) --  (6*\XA,6*\YA) --  (6*\XA,2*\YA)  --  (6.5*\XA,2*\YA){};

\draw [very thick,red]  (2*\XA,5*\YA) -- (5*\XA,5*\YA) -- (5*\XA,1*\YA) -- (6.5*\XA,1*\YA){};

\draw [very thick,red] (8.5*\XA,2*\YA) --  (9*\XA,2*\YA){};
\draw [very thick,blue] (9*\XA,2*\YA) --  (10*\XA,2*\YA) -- (10.5*\XA,1.5*\YA){};

\draw [very thick,red] (8.5*\XA,1*\YA) -- (9*\XA,1*\YA){};

\draw [very thick,red] (11*\XA,0*\YA) --  (11*\XA,-0.5*\YA){};
\draw [very thick,blue] (10.5*\XA,1.5*\YA) -- (11*\XA,1*\YA) --  (11*\XA,0*\YA){};

\draw [very thick,red]   (10*\XA,0*\YA) --  (10*\XA,-0.5*\YA){};

\draw [very thick,red]  (10*\XA,-2.5*\YA) -- (10*\XA,-3*\YA){};
\draw [very thick,red]  (11*\XA,-2.5*\YA) -- (11*\XA,-3*\YA){};

\draw [very thick,red]  (8.5*\XA,-5*\YA) -- (9*\XA,-5*\YA) {};
\draw [very thick,blue]  (9*\XA,-5*\YA) -- (10*\XA,-5*\YA) -- (10*\XA,-4*\YA) -- (11*\XA,-3*\YA) {};

\draw [very thick,red]    (9*\XA,-4*\YA) -- (8.5*\XA,-4*\YA){};

\draw [very thick,red]    (6.5*\XA,-4*\YA) -- (1*\XA,-4*\YA) -- (1*\XA,-2.5*\YA){};
\draw [very thick,red]    (1*\XA,-0.5*\YA) -- (1*\XA,4*\YA) {};

\draw [very thick,red]    (6.5*\XA,-5*\YA) -- (0*\XA,-5*\YA) -- (0*\XA,-2.5*\YA){};

\draw [very thick,red]    (0*\XA,-0.5*\YA) -- (0*\XA,4*\YA) {};

\draw [very thick,blue]  (0*\XA,4*\YA) -- (0*\XA,5*\YA) -- (0.5*\XA,5.5*\YA){};

\foreach \i in {-2,-1,0,2,3,4,8,9,10}{
\foreach \j in {5}{
\draw [very thin,black] (\i*\XA+1.5*\XA,\j*\YA+1.5*\XA) -- (\i*\XA+1*\XA,\j*\YA+1*\XA){};
\draw [very thin,black] (\i*\XA+1.5*\XA,\j*\YA+1.5*\XA) -- (\i*\XA+2*\XA,\j*\YA+1*\XA){};
\draw [very thin,black] (\i*\XA+1.5*\XA,\j*\YA+1.5*\XA) -- (\i*\XA+1*\XA,\j*\YA+2*\XA){};
\draw [very thin,black] (\i*\XA+1.5*\XA,\j*\YA+1.5*\XA) -- (\i*\XA+2*\XA,\j*\YA+2*\XA){};
\node[roundnode2] (a0) at (\i*\XA+1.5*\XA,\j*\YA+1.5*\XA) {}; 
\node[roundnode3] (a0) at (\i*\XA+1.75*\XA,\j*\YA+1.75*\XA) {}; 
\node[roundnode3] (a0) at (\i*\XA+1.75*\XA,\j*\YA+1.25*\XA) {}; 
\node[roundnode3] (a0) at (\i*\XA+1.25*\XA,\j*\YA+1.25*\XA) {}; 
\node[roundnode3] (a0) at (\i*\XA+1.25*\XA,\j*\YA+1.75*\XA) {}; 
}}

\foreach \i in {2,3,4,8,9,10}{
\foreach \j in {3,4}{
\draw [very thin,black] (\i*\XA+1.5*\XA,\j*\YA+1.5*\XA) -- (\i*\XA+1*\XA,\j*\YA+1*\XA){};
\draw [very thin,black] (\i*\XA+1.5*\XA,\j*\YA+1.5*\XA) -- (\i*\XA+2*\XA,\j*\YA+1*\XA){};
\draw [very thin,black] (\i*\XA+1.5*\XA,\j*\YA+1.5*\XA) -- (\i*\XA+1*\XA,\j*\YA+2*\XA){};
\draw [very thin,black] (\i*\XA+1.5*\XA,\j*\YA+1.5*\XA) -- (\i*\XA+2*\XA,\j*\YA+2*\XA){};
\node[roundnode2] (a0) at (\i*\XA+1.5*\XA,\j*\YA+1.5*\XA) {}; 
\node[roundnode3] (a0) at (\i*\XA+1.75*\XA,\j*\YA+1.75*\XA) {}; 
\node[roundnode3] (a0) at (\i*\XA+1.75*\XA,\j*\YA+1.25*\XA) {}; 
\node[roundnode3] (a0) at (\i*\XA+1.25*\XA,\j*\YA+1.25*\XA) {}; 
\node[roundnode3] (a0) at (\i*\XA+1.25*\XA,\j*\YA+1.75*\XA) {}; 
}}

\foreach \i in {-2,-1,0,2,3,4}{
\foreach \j in {1,0,-1,-5,-6,-7}{
\draw [very thin,black] (\i*\XA+1.5*\XA,\j*\YA+1.5*\XA) -- (\i*\XA+1*\XA,\j*\YA+1*\XA){};
\draw [very thin,black] (\i*\XA+1.5*\XA,\j*\YA+1.5*\XA) -- (\i*\XA+2*\XA,\j*\YA+1*\XA){};
\draw [very thin,black] (\i*\XA+1.5*\XA,\j*\YA+1.5*\XA) -- (\i*\XA+1*\XA,\j*\YA+2*\XA){};
\draw [very thin,black] (\i*\XA+1.5*\XA,\j*\YA+1.5*\XA) -- (\i*\XA+2*\XA,\j*\YA+2*\XA){};
\node[roundnode2] (a0) at (\i*\XA+1.5*\XA,\j*\YA+1.5*\XA) {}; 
\node[roundnode3] (a0) at (\i*\XA+1.75*\XA,\j*\YA+1.75*\XA) {}; 
\node[roundnode3] (a0) at (\i*\XA+1.75*\XA,\j*\YA+1.25*\XA) {}; 
\node[roundnode3] (a0) at (\i*\XA+1.25*\XA,\j*\YA+1.25*\XA) {}; 
\node[roundnode3] (a0) at (\i*\XA+1.25*\XA,\j*\YA+1.75*\XA) {}; 
}}

\foreach \i in {9,10}{
\foreach \j in {1,-7}{
\draw [very thin,black] (\i*\XA+1.5*\XA,\j*\YA+1.5*\XA) -- (\i*\XA+1*\XA,\j*\YA+1*\XA){};
\draw [very thin,black] (\i*\XA+1.5*\XA,\j*\YA+1.5*\XA) -- (\i*\XA+2*\XA,\j*\YA+1*\XA){};
\draw [very thin,black] (\i*\XA+1.5*\XA,\j*\YA+1.5*\XA) -- (\i*\XA+1*\XA,\j*\YA+2*\XA){};
\draw [very thin,black] (\i*\XA+1.5*\XA,\j*\YA+1.5*\XA) -- (\i*\XA+2*\XA,\j*\YA+2*\XA){};
\node[roundnode2] (a0) at (\i*\XA+1.5*\XA,\j*\YA+1.5*\XA) {}; 
\node[roundnode3] (a0) at (\i*\XA+1.75*\XA,\j*\YA+1.75*\XA) {}; 
\node[roundnode3] (a0) at (\i*\XA+1.75*\XA,\j*\YA+1.25*\XA) {}; 
\node[roundnode3] (a0) at (\i*\XA+1.25*\XA,\j*\YA+1.25*\XA) {}; 
\node[roundnode3] (a0) at (\i*\XA+1.25*\XA,\j*\YA+1.75*\XA) {}; 
}}

\foreach \i in {9}{
\foreach \j in {-5}{
\draw [very thin,black] (\i*\XA+1.5*\XA,\j*\YA+1.5*\XA) -- (\i*\XA+1*\XA,\j*\YA+1*\XA){};
\draw [very thin,black] (\i*\XA+1.5*\XA,\j*\YA+1.5*\XA) -- (\i*\XA+2*\XA,\j*\YA+1*\XA){};
\draw [very thin,black] (\i*\XA+1.5*\XA,\j*\YA+1.5*\XA) -- (\i*\XA+1*\XA,\j*\YA+2*\XA){};
\draw [very thin,black] (\i*\XA+1.5*\XA,\j*\YA+1.5*\XA) -- (\i*\XA+2*\XA,\j*\YA+2*\XA){};
\node[roundnode2] (a0) at (\i*\XA+1.5*\XA,\j*\YA+1.5*\XA) {}; 
\node[roundnode3] (a0) at (\i*\XA+1.75*\XA,\j*\YA+1.75*\XA) {}; 
\node[roundnode3] (a0) at (\i*\XA+1.75*\XA,\j*\YA+1.25*\XA) {}; 
\node[roundnode3] (a0) at (\i*\XA+1.25*\XA,\j*\YA+1.25*\XA) {}; 
\node[roundnode3] (a0) at (\i*\XA+1.25*\XA,\j*\YA+1.75*\XA) {}; 
}}

\foreach \i in {10}{
\foreach \j in {-1,-6,-5}{
\draw [very thin,black] (\i*\XA+1.5*\XA,\j*\YA+1.5*\XA) -- (\i*\XA+1*\XA,\j*\YA+1*\XA){};
\draw [very thin,black] (\i*\XA+1.5*\XA,\j*\YA+1.5*\XA) -- (\i*\XA+2*\XA,\j*\YA+1*\XA){};
\draw [very thin,black] (\i*\XA+1.5*\XA,\j*\YA+1.5*\XA) -- (\i*\XA+1*\XA,\j*\YA+2*\XA){};
\draw [very thin,black] (\i*\XA+1.5*\XA,\j*\YA+1.5*\XA) -- (\i*\XA+2*\XA,\j*\YA+2*\XA){};
\node[roundnode2] (a0) at (\i*\XA+1.5*\XA,\j*\YA+1.5*\XA) {}; 
\node[roundnode3] (a0) at (\i*\XA+1.75*\XA,\j*\YA+1.75*\XA) {}; 
\node[roundnode3] (a0) at (\i*\XA+1.75*\XA,\j*\YA+1.25*\XA) {}; 
\node[roundnode3] (a0) at (\i*\XA+1.25*\XA,\j*\YA+1.25*\XA) {}; 
\node[roundnode3] (a0) at (\i*\XA+1.25*\XA,\j*\YA+1.75*\XA) {}; 
}}

\foreach \i in {-1,3,9}{
\foreach \j in {0,4,-6}{
\draw[dashed,draw=black,fill=none,thick] (\i*\XA-0.2*\XA,\j*\YA-0.2*\YA) rectangle (\i*\XA+3.2*\XA,\j*\YA+3.2*\YA) ;
\draw [very thin,black] (\i*\XA+1.5*\XA,\j*\YA+1.5*\XA) -- (\i*\XA+1*\XA,\j*\YA+1*\XA){};
\draw [very thin,black] (\i*\XA+1.5*\XA,\j*\YA+1.5*\XA) -- (\i*\XA+2*\XA,\j*\YA+1*\XA){};
\draw [very thin,black] (\i*\XA+1.5*\XA,\j*\YA+1.5*\XA) -- (\i*\XA+1*\XA,\j*\YA+2*\XA){};
\draw [very thin,black] (\i*\XA+1.5*\XA,\j*\YA+1.5*\XA) -- (\i*\XA+2*\XA,\j*\YA+2*\XA){};
\node[roundnode2] (a0) at (\i*\XA+1.5*\XA,\j*\YA+1.5*\XA) {}; 
\node[roundnode3] (a0) at (\i*\XA+1.75*\XA,\j*\YA+1.75*\XA) {}; 
\node[roundnode3] (a0) at (\i*\XA+1.75*\XA,\j*\YA+1.25*\XA) {}; 
\node[roundnode3] (a0) at (\i*\XA+1.25*\XA,\j*\YA+1.25*\XA) {}; 
\node[roundnode3] (a0) at (\i*\XA+1.25*\XA,\j*\YA+1.75*\XA) {}; 
}}

\draw [very thick,black,->] (-0.3*\XA,5.5*\YA) .. controls (0*\XA,5.7*\YA) .. (0.3*\XA,5.5*\YA){};
\node[draw=none,fill=none,] at (-0.5*\XA,5.45*\XA) {$u$};

\draw [very thick,black,->] (1.6*\XA,5.6*\YA) -- (1.9*\XA,5.9*\YA){};
\node[draw=none,fill=none,] at (1.5*\XA,5.45*\XA) {$u_1$};
\draw [very thick,black,->] (1.6*\XA,4.6*\YA) -- (1.9*\XA,4.9*\YA){};
\node[draw=none,fill=none,] at (1.5*\XA,4.45*\XA) {$u_2$};

\draw [very thick,black,->] (-0.4*\XA,4.4*\YA) -- (-0.1*\XA,4.1*\YA){};
\node[draw=none,fill=none,] at (-0.6*\XA,4.5*\XA) {$u'_2$};
\draw [very thick,black,->] (0.6*\XA,4.4*\YA) -- (0.9*\XA,4.1*\YA){};
\node[draw=none,fill=none,] at (0.4*\XA,4.5*\XA) {$u'_1$};

\draw [very thick,black,->] (9.4*\XA,2.4*\YA) -- (9.1*\XA,2.1*\YA){};
\node[draw=none,fill=none,] at (9.5*\XA,2.5*\XA) {$v_1$};
\draw [very thick,black,->] (9.4*\XA,1.4*\YA) -- (9.1*\XA,1.1*\YA){};
\node[draw=none,fill=none,] at (9.5*\XA,1.5*\XA) {$v_2$};

\draw [very thick,black,->] (9.65*\XA,0.35*\YA) -- (9.9*\XA,0.1*\YA){};
\node[draw=none,fill=none,] at (9.5*\XA,0.55*\XA) {$v'_1$};
\draw [very thick,black,->] (10.65*\XA,0.35*\YA) -- (10.9*\XA,0.1*\YA){};
\node[draw=none,fill=none,] at (10.5*\XA,0.55*\XA) {$v'_2$};

\draw [very thick,black,->] (9.65*\XA,-2.5*\YA) -- (9.9*\XA,-2.9*\YA){};
\node[draw=none,fill=none,] at (9.5*\XA,-2.3*\XA) {$w_1$};
\draw [very thick,black,->] (10.65*\XA,-2.5*\YA) -- (10.9*\XA,-2.9*\YA){};
\node[draw=none,fill=none,] at (10.5*\XA,-2.3*\XA) {$w_2$};

\draw [very thick,black,->] (9.35*\XA,-4.4*\YA) -- (9.1*\XA,-4.1*\YA){};
\node[draw=none,fill=none,] at (9.5*\XA,-4.6*\XA) {$w'_1$};
\draw [very thick,black,->] (9.35*\XA,-5.4*\YA) -- (9.1*\XA,-5.1*\YA){};
\node[draw=none,fill=none,] at (9.5*\XA,-5.6*\XA) {$w'_2$};

\draw [very thick,black,<-] (10.7*\XA,1.5*\YA) .. controls (11*\XA,1.7*\YA) .. (11.2*\XA,1.5*\YA){};
\node[draw=none,fill=none,] at (11.3*\XA,1.4*\XA) {$v$};

\draw [very thick,black,->] (9.7*\XA,-3.5*\YA) .. controls (10*\XA,-3.3*\YA) .. (10.3*\XA,-3.5*\YA){};
\node[draw=none,fill=none,] at (9.5*\XA,-3.6*\XA) {$w$};

\foreach \i in {-6,-5,-4,-3,0,1,2,3,4,5,6,7} {
        \draw [very thin,black] (-1*\XA,\i*\YA) -- (6.5*\XA,\i*\YA){};
}
\foreach \i in {-6,-5,-4,-3,0,1,2,3,4,5,6,7} {
        \draw [very thin,black] (8.5*\XA,\i*\YA) -- (12*\XA,\i*\YA){};
}     
     
\foreach \i in {-1,...,6,9,10,11,12} {     
        \draw [very thin,black] (\i*\XA,-0.5*\YA) -- (\i*\XA,7*\YA){};
}

\foreach \i in {-1,...,6,9,10,11,12} {     
        \draw [very thin,black] (\i*\XA,-6*\YA) -- (\i*\XA,-2.5*\YA){};
}

\foreach \i in {-1,...,6,9,10,11,12} {     
        \draw [dotted,very thick,black] (\i*\XA,-2*\YA) -- (\i*\XA,-1*\YA){};
}

\foreach \i in {-6,-5,-4,-3,0,1,2,3,4,5,6,7} {
        \draw [dotted,very thick,black](7*\XA,\i*\YA) -- (8*\XA,\i*\YA){};
}    

\foreach \i in {-1,0,1,2,3,4,5,6,9,10,11,12} {
	\foreach \j in {-6,-5,-4,-3,0,1,2,3,4,5,6,7} {
			\node[roundnode] (a0) at (\i*\XA,\j*\YA) {};   
	}
}

\foreach \i in {-0.5,0.5,1.5,2.5,3.5,4.5,5.5,9.5,10.5,11.5} {
	\foreach \j in {-6,-5,-4,-3,0,1,2,3,4,5,6,7} {
			\node[roundnode3] (a0) at (\i*\XA,\j*\YA) {};   
	}
}

\foreach \i in {-1,0,1,2,3,4,5,6,9,10,11,12} {
	\foreach \j in {-5.5,-4.5,-3.5,0.5,1.5,2.5,3.5,4.5,5.5,6.5} {
			\node[roundnode3] (a0) at (\i*\XA,\j*\YA) {};   
	}
}

\draw [thin,black,<->] (-1.1*\XA,7.5*\YA) -- (2.1*\XA,7.5*\YA){};
\draw [thin,black,<->] (-1.1*\XA,8.3*\YA) -- (12.1*\XA,8.3*\YA){};
\draw [thin,black,<->] (2.9*\XA,7.5*\YA) -- (6.1*\XA,7.5*\YA){};
\draw [thin,black,<->] (8.9*\XA,7.5*\YA) -- (12.1*\XA,7.5*\YA){};
\node[draw=none,fill=none,] at (0.5*\XA,7.9*\YA) {$d$};
\node[draw=none,fill=none,] at (10.5*\XA,7.9*\YA) {$d$};
\node[draw=none,fill=none,] at (4.5*\XA,7.9*\YA) {$d$};
\node[draw=none,fill=none,] at (5.1*\XA,8.6*\YA) {$d^3$};

\draw [thin,black,<->] (-1.6*\XA,3.9*\YA) -- (-1.6*\XA,7.1*\YA){};
\draw [thin,black,<->] (-1.6*\XA,-0.1*\YA) -- (-1.6*\XA,3.1*\YA){};
\draw [thin,black,<->] (-1.6*\XA,-6.1*\YA) -- (-1.6*\XA,-2.9*\YA){};
\draw [thin,black,<->] (-2.4*\XA,-6.1*\YA) -- (-2.4*\XA,7.1*\YA){};
\node[draw=none,fill=none,] at (-2*\XA,5.5*\YA) {$d$};
\node[draw=none,fill=none,] at (-2*\XA,1.5*\YA) {$d$};
\node[draw=none,fill=none,] at (-2*\XA,-4.5*\YA) {$d$};
\node[draw=none,fill=none,] at (-2.9*\XA,0.5*\YA) {$d^3$};

\end{tikzpicture}
\caption{The paths considered to construct the cycle $\Cc$ in the proof of Proposition \ref{prop1} are depicted in red. The blue line show how the paths are completed to construct the cycle $\Cc$ in the case $i_1=1$, $i_2=2$ and $i_3=2$. Note that some vertices in $V_d^2$ are omitted in order to make the figure clearer.}\label{fig5}
\end{figure}
\subsection{Upper bound on the success probability}\label{subseqproof3}
Let $(u,v,w)$ denote the triple from $\Uu\times \Vv\times \Ww$ whose existence is guaranteed by Proposition \ref{prop1}. Let $\Cc$, $q_1$, $q_2$ and $q_3$ denote the cycle and the three paths of Condition (ii) of the proposition.

Remember that each entry $A_{ij}$ of the input matrix~$A$ specifies the basis in which the qubit of vertex~$u_{ij}$ in the graph state $\ket{\overline{\Gg}_d}$ is measured. We will say that the vertex $u_{ij}$ is marked if $A_{ij}=1$.
The input matrix $A\in \{0,1\}^{k\times k}$ can then be constructed by first considering the $k^2-3$ entries corresponding to all vertices in $V_d^2\setminus\{u,v,w\}$, and then specifying the entries of the three vertices $u$, $v$ and $w$. This means that $A$ can be represented as a pair of strings $(a,b)$ where $a\in\{0,1\}^{k^2-3}$ and $b=(b_u,b_v,b_w)\in\{0,1\}^3$. 

The randomized classical circuit $C_d$ can be seen as a deterministic circuit receiving a random string~$r$. Let us fix the value of this random string. Let us also fix the string $a\in\{0,1\}^{k^2-3}$ and assume that the Hamming weight $|a|$ is even (note that $|a|$ corresponds to the number of marked vertices in $V_d^2\setminus \{u,v,w\}$). The only remaining variables are thus the three bits $b_u$, $b_v$ and $b_w$.

Observe that the graph $\overline{\Gg}_{d}$ remains connected when removing all the vertices on the cycle~$\Cc$, due to Conditions (ii-a) and (ii-b) of Proposition \ref{prop1}. No vertex from $V_d^2\setminus \{u,v,w\}$ appears in $\Cc$, from Condition~(ii-b) of Proposition \ref{prop1}. 
This implies that there exists a set of $|a|/2$ paths $\{p_1,\ldots,p_{|a|/2}\}$ such that
$p_i\cap \Cc=\emptyset$ for each $i\in\{1,\ldots,|a|/2\}$, and
each marked vertex in $V_d^2\setminus \{u,v,w\}$ appears once as an endpoint of one of these paths.
Define the three bits
\begin{eqnarray*}
\lambda_1&=&\bigoplus_{\ell\in V_d} z_{\ell},\\
\lambda_2&=&\bigoplus_{i=1}^{|a|/2}\bigoplus_{\ell\in p_i\cap V_d^\ast} z_{\ell},\\
y&=&
\left\{
\begin{tabular}{ll}
$\lambda_1\oplus \lambda_2$&\textrm{ if }$|a|\bmod 4=0$,\\
$\lambda_1\oplus \lambda_2\oplus 1$&\textrm{ if }$|a|\bmod 4=2$.
\end{tabular}
\right.
\end{eqnarray*}
A crucial observation is that $y$ is an affine function of $b_u$, $b_v$ and $b_w$, due to Condition (i) of Proposition~\ref{prop1}.

Define
\[
y_1=\bigoplus_{\ell\in q_1\cap V_d^\ast} z_\ell,\hspace{5mm}
y_2=\bigoplus_{\ell\in q_2\cap V_d^\ast} z_\ell,\hspace{5mm}
y_3=\bigoplus_{\ell\in q_3\cap V_d^\ast} z_\ell.
\]
Condition (i) of Proposition~\ref{prop1} again guarantees that $y_1$, $y_2$ and $y_3$ are affine functions of the three bits $b_u$, $b_v$, $b_w$. Moreover, Condition (ii-c) implies that $y_1$ does not depend on $b_u$, $y_2$ does not depend on $b_v$ and~$y_3$ does not depend on $b_w$.

Theorem \ref{th:cond1} implies that if the output of the circuit is in the set $\Lambda_d(A)$ (i.e., the output corresponds to a valid outcome arising from the corresponding measurement of the graph state $\ket{\overline{\Gg}_d}$), then the following condition should hold:
\begin{equation}\label{eqA1}
y_1\oplus y_2\oplus y_3=0 \hspace{3mm} \textrm{ for all } (b_u,b_v,b_w)\in\{0,1\}^3.
\end{equation}
Theorem \ref{th:cond2} additionally implies that if the output of the circuit is in the set $\Lambda_d(A)$ then the following condition should hold:
\begin{equation}\label{eqA2}
\left\{
\begin{tabular}{ll}
$y = 0$ &\textrm{ if } $(b_u,b_v,b_w)=(0,0,0)$,\\
$y \oplus y_1 = 1$ &\textrm{ if } $(b_u,b_v,b_w)=(0,1,1)$,\\
$y \oplus y_2 = 1$ &\textrm{ if } $(b_u,b_v,b_w)=(1,0,1)$,\\
$y \oplus y_3 = 1$ &\textrm{ if } $(b_u,b_v,b_w)=(1,1,0)$.
\end{tabular}
\right.
\end{equation}
Lemma \ref{lemma1} implies that there is at least one value for the triple $(b_u,b_v,b_w)$ for which these conditions are not satisfied.
 
We have just shown that for any value of $r$ and any value of $a$ such that $|a|$ is even, the output of the circuit $C_d$ is incorrect for at least a fraction $1/8$ of the strings $b=(b_u,b_v,b_w)\in\{0,1\}^{3}$. Since $|a|$ is even with probability $1/2$ when choosing the matrix $A$ uniformly at random, we conclude that for any value of~$r$ the output of the circuit is incorrect for at least a fraction $1/16$ of the matrices $A\in\{0,1\}^{k\times k}$. This implies the inequality
\[
\sum_{A\in\{0,1\}^{k\times k}}\Pr_r[C_d(A)\notin \Lambda_d(A)]\ge \frac{2^{k^2}}{16}.
\]
and thus
\[
\frac{1}{2^{k^2}}\sum_{A\in\{0,1\}^{k\times k}}\Pr_r[C_d(A)\in \Lambda_d(A)]< 15/16.
\]
This concludes the proof of Theorem \ref{th:main}.

\section{Soundness Amplification for Small-Depth Circuits}\label{sec:amp}
In this section we show how to obtain Theorem \ref{th:main-amp} from Theorem \ref{th:main}. In Section \ref{sub:amp-gen} we first present a general soundness amplification result that holds for any relation. Then in Section \ref{sub:amp-app} we apply this result to the relation $R_d$ of Theorem \ref{th:main} in order to obtain Theorem \ref{th:main-amp}. 

\subsection{General result}\label{sub:amp-gen}
Consider any relation $\Rgen\subseteq \{0,1\}^m\times \{0,1\}^n$ for some positive integers $m$ and $n$. As usual, this relation is interpreted as the following computational problem: given as input a string $x\in \{0,1\}^m$, output one string from the set $\Rgen(x)=\{z\in\{0,1\}^n\:|\:(x,z)\in\Rgen\}$. 
For any integer $t\ge 1$, now consider the following computational problem: given as input $t$ strings $x_1,\ldots,x_t\in \{0,1\}^m$, output one element from the set $\Rgen(x_1)\times \cdots\times \Rgen(x_t)$. This computational problem corresponds to the direct product of $t$ copies of the relation $\Rgen$. We will write this relation $\Rgen^{\times t}$ and interpret it as the subset
\[
\Rgen^{\times t}\subseteq \{0,1\}^{mt}\times \{0,1\}^{nt}
\]
by associating $\{0,1\}^{mt}$ with the $t$ copies of  $\{0,1\}^{m}$ and $\{0,1\}^{nt}$ with the $t$ copies of $\{0,1\}^{n}$.

The main result of this section is the following repetition theorem, which shows that if $\Rgen$ cannot be computed with average success probability larger than $1-\alpha$ using small-depth classical circuits, then $\Rgen^{\times t}$ cannot be computed with average success probability larger than $(1-\alpha)^{t'}$ for some $t'\approx t$ by circuits of the same depth. The idea is to show how to extract, from the $t$ copies of $\Rgen$ making $\Rgen^{\times t}$, at least~$t'$ copies on which the circuit acts independently. 
\begin{theorem}\label{th:parallel}
Let $\Rgen\subseteq \{0,1\}^m\times \{0,1\}^n$ be a relation for which the following assertion holds for some real numbers $c\ge 0$ and $\alpha\in[0,1]$: any $m$-input $n$-output randomized circuit $C$ with bounded-fanin gates and depth at most $c\log_2 m$ satisfies the inequality 
\[
\frac{1}{2^m}\sum_{x\in\{0,1\}^m}\Pr[C(x)\in \Rgen(x)]< 1-\alpha.
\]
Let $t$ be any integer such that $t\ge 6nm^c+2$.
Then any $(mt)$-input $(nt)$-output randomized circuit $C'$ with bounded-fanin gates and depth at most $c\log_2 m$ satisfies
\[
\frac{1}{2^{mt}}\sum_{x'\in\{0,1\}^{mt}}\Pr[C'(x')\in \Rgen^{\times t}(x')] < (1-\alpha)^{t/(6m^{c}n+2)}.
\]
\end{theorem}
\begin{proof}
Consider any $(mt)$-input $(nt)$-output randomized circuit $C'$ with gates of fanin at most 2 and depth at most $c\log_2 m$ for the relation $\Rgen^{\times t}$. For each $i\in\{1,\ldots,t\}$, let $S_i$ denote the set of wires corresponding to the inputs of the $i$-th copy of $\Rgen$ in $\Rgen^{\times t}$ and $T_i$ denote the set of wires corresponding to the outputs of the $i$-th copy of $\Rgen$ in $\Rgen^{\times t}$. The following claim is the crucial part of the proof.
\begin{claim}\label{claim:I}
There exists a subset of indices $I\subseteq \{1,\ldots,t\}$ of size 
$|I|\ge \frac{t}{6nm^c+2}$
such that $L(S_i)\cap T_j=\emptyset$ for all distinct $i,j\in I$.
\end{claim}
\begin{proof}
Define the set
\begin{align*}
J =& \Big\{
i\in \{1,\ldots,t\}\:|\: \sum_{x\in S_i}|L(x)|\le 2m^{c}n
\Big\}.
\end{align*}
Since the circuit $C'$ has depth $c\log_2 m$ and its gates have fanin at most 2, we have $|L(z)|\le m^c$ for any output wire $z$. Since the total number of output wires is $nt$, a simple counting argument shows that $|J|\ge t/2$.

Let us now construct a graph on the vertex set $J$ as follows: two distinct vertices $i,j\in J$ are connected by an edge if and only if at least one of $L(S_i)\cap T_j\neq\emptyset$ and $L(S_j)\cap T_i\neq \emptyset$ holds. In this graph each vertex has degree at most $2m^{c}n+m^{c}n=3m^{c}n$. There thus exists\footnote{Here we are using a trivial result from graph theory that states that a graph of maximum degree $\Delta$ on $N$ vertices has an independent set of size at least $N/(\Delta+1)$.} an independent set  $I\subseteq J$ of $G$ of size
\[
|I|\ge \frac{t/2}{3m^{c}n+1}=\frac{t}{6m^{c}n+2}.
\]
This independent set is precisely the set of indices we wanted to construct. 
\end{proof}

To lighten the notation we will assume that the set $I$ from Claim~\ref{claim:I} is $I=\{1,\ldots,\ell\}$ for some integer~$\ell$ (with $\ell\ge \frac{t}{6m^{c}n+2}$). This assumption can be made without loss of generality. Claim \ref{claim:I} implies that when the values of all the input wires in $S_{\ell+1}\cup\cdots\cup S_t$ are fixed, then for each $i\in\{1,\ldots, \ell\}$ the values of the output wires in $T_i$ only depend on the values of the input wires in $S_i$.
This implies that for any $(x_{\ell+1},\ldots,x_t)\in \{0,1\}^{(t-\ell)m}$ the inequality 
\begin{align*}
\frac{1}{2^{m\ell}}\sum_{x_1,\ldots,x_\ell\in\{0,1\}^{m}}\Pr[C'(x_1,\ldots,x_\ell,x_{\ell+1},\ldots,x_t)\in \Rgen^{\times t}(x_1,\ldots,x_\ell,x_{\ell+1},\ldots,x_t)]
< (1-\alpha)^{\ell},
\end{align*}
holds, from our assumption on the relation $\Rgen$ (since the depth of $C'$ is at most $c\log_2 m$). Thus
\begin{align*}
\frac{1}{2^{mt}}\sum_{x_1,\ldots,x_t\in\{0,1\}^{m}}\Pr[C'(x_1,\ldots,x_t)\in \Rgen^{\times t}(x_1,\ldots,x_t)]
< (1-\alpha)^{\ell}\le (1-\alpha)^{t/(6m^{c}n+2)},
\end{align*}
as claimed. This concludes the proof of Theorem \ref{th:parallel}.
\end{proof}

\subsection{Application: proof of Theorem \ref{th:main-amp}}\label{sub:amp-app}
We are now able to give the proof of Theorem \ref{th:main-amp}.
\begin{proof}[Proof of Theorem \ref{th:main-amp}]
We consider the relation $R_d\subseteq \{0,1\}^m\times \{0,1\}^n$ defined in Section \ref{sub:def} and used in Theorem \ref{th:main}. Remember that for this relation we have $m=\Theta(d^6)$ and $n=\Theta(d^6)$. Take the integer $t=\ceil{(6nm^{1/8}+2)^3}$ and observe that the inequality $t\ge m^{27/8}$ holds. Define $\Ramp=R_d^{\times t}$. The sizes of the inputs and outputs in $\Ramp$ are $M=mt$ and $N=nt$, respectively. Observe that $t\ge m^{27/8}$ implies $t\ge M^{27/35}$.
Theorem~\ref{th:main} and then Theorem \ref{th:parallel} with $\Rgen=R_d$ imply that there exist constants $c> 0$ and $\alpha>0$ such that
any $M$-input $N$-output randomized circuit $C'$ with bounded-fanin gates and depth at most $c\log_2 m$ satisfies
\[
\frac{1}{2^{M}}\sum_{x'\in\{0,1\}^{M}}\Pr[C'(x')\in \Ramp(x')] <\left(1-\alpha\right)^{t^{2/3}}\le\left(1-\alpha\right)^{M^{54/105}}<\left(1-\alpha\right)^{\sqrt{M}},
\] 
which leads to the claimed statement.
\end{proof}

\section*{Acknowledgments}
The author is grateful to Keisuke Fujii, Tomoyuki Morimae, Harumichi Nishimura, Ansis Rosmanis and Yasuhiro Takahashi for helpful discussions. The author also thanks Gabriel Senno, Jalex Stark, Thomas Vidick and anonymous reviewers for comments about the manuscript.
This work was partially supported by the JSPS KAKENHI grants No.~15H01677, No.~16H01705 and No.~16H05853. 


\begin{thebibliography}{GHMP02}

\bibitem[AA11]{Aaronson+STOC11}
Scott Aaronson and Alex Arkhipov.
\newblock The computational complexity of linear optics.
\newblock In {\em Proceedings of the 43rd {ACM} Symposium on Theory of
  Computing}, pages 333--342, 2011.

\bibitem[AA14]{Aaronson+A14}
Scott Aaronson and Alex Arkhipov.
\newblock {BosonSampling} is far from uniform.
\newblock {\em Quantum Information {\&} Computation}, 14(15-16):1383--1423,
  2014.

\bibitem[AC17]{Aaronson+CCC17}
Scott Aaronson and Lijie Chen.
\newblock Complexity-theoretic foundations of quantum supremacy experiments.
\newblock In {\em Proceedings of the 32nd Computational Complexity Conference},
  pages 22:1--22:67, 2017.

\bibitem[Amb18]{Ambainis18}
Andris Ambainis.
\newblock Understanding quantum algorithms via query complexity.
\newblock In {\em Proceedings of the 2018 International Congress of
  Mathematicians}, volume~3, pages 3249--3270, 2018.

\bibitem[BCE{\etalchar{+}}07]{Barrett+PRA07}
Jonathan Barrett, Carlton~M. Caves, Bryan Eastin, Matthew~B. Elliott, and
  Stefano Pironio.
\newblock Modeling {Pauli} measurements on graph states with nearest-neighbor
  classical communication.
\newblock {\em Physical Review A}, 75:012103, 2007.

\bibitem[BFNV19]{Bouland+18}
Adam Bouland, Bill Fefferman, Chinmay Nirkhe, and Umesh Vazirani.
\newblock ``{Quantum} supremacy'' and the complexity of random circuit
  sampling.
\newblock In {\em Proceedings of the 10th Innovations in Theoretical Computer
  Science conference}, pages 15:1--15:2, 2019.
\newblock arXiv:1803.04402.

\bibitem[BGK17]{Bravyi+17}
Sergey Bravyi, David Gosset, and Robert K\"onig.
\newblock Quantum advantage with shallow circuits.
\newblock arXiv:1704.00690 (preliminary version of \cite{Bravyi+18}), 2017.

\bibitem[BGK18]{Bravyi+18}
Sergey Bravyi, David Gosset, and Robert K\"onig.
\newblock Quantum advantage with shallow circuits.
\newblock {\em Science}, 362(6412):308--311, 2018.

\bibitem[BJS10]{Bremner+10}
Michael~J. Bremner, Richard Jozsa, and Dan~J. Shepherd.
\newblock Classical simulation of commuting quantum computations implies
  collapse of the polynomial hierarchy.
\newblock {\em Proceedings of the Royal Society of London A: Mathematical,
  Physical and Engineering Sciences}, 467(2126):459--472, 2010.

\bibitem[BKST19]{Bene+18}
Adam {Bene Watts}, Robin Kothari, Luke Schaeffer, and Avishay Tal.
\newblock Exponential separation between shallow quantum circuits and unbounded
  fan-in shallow classical circuits.
\newblock In {\em Proceedings of the 43rd {ACM} Symposium on Theory of
  Computing}, 2019.
\newblock To appear.

\bibitem[BMS16]{Bremner+PRL16}
Michael~J. Bremner, Ashley Montanaro, and Dan~J. Shepherd.
\newblock Average-case complexity versus approximate simulation of commuting
  quantum computations.
\newblock {\em Physical Review Letters}, 117:080501, 2016.

\bibitem[BMS17]{Bremner+17}
Michael~J. Bremner, Ashley Montanaro, and Dan~J. Shepherd.
\newblock Achieving quantum supremacy with sparse and noisy commuting quantum
  circuits.
\newblock {\em Quantum}, 1:8, 2017.

\bibitem[BV97]{Bernstein+SICOMP97}
Ethan Bernstein and Umesh~V. Vazirani.
\newblock Quantum complexity theory.
\newblock {\em {SIAM} Journal on Computing}, 26(5):1411--1473, 1997.

\bibitem[CSV18]{Coudron+18}
Matthew Coudron, Jalex Stark, and Thomas Vidick.
\newblock Trading locality for time: certifiable randomness from low-depth
  circuits.
\newblock arXiv:1810.04233, 2018.

\bibitem[FH16]{Fahri+16}
Edward Farhi and Aram~W. Harrow.
\newblock Quantum supremacy through the quantum approximate optimization
  algorithm.
\newblock arXiv:1602.07674, 2016.

\bibitem[FKM{\etalchar{+}}18]{Fujii+PRL18}
Keisuke Fujii, Hirotada Kobayashi, Tomoyuki Morimae, Harumichi Nishimura,
  Shuhei Tamate, and Seiichiro Tani.
\newblock Impossibility of classically simulating one-clean-qubit model with
  multiplicative error.
\newblock {\em Physical Review Letters}, 120:200502, 2018.

\bibitem[FT16]{Fujii+16}
Keisuke Fujii and Shuhei Tamate.
\newblock Computational quantum-classical boundary of noisy commuting quantum
  circuits.
\newblock {\em Scientific Reports}, 6(25598), 2016.

\bibitem[GHMP02]{Green+02}
Frederic Green, Steven Homer, Cristopher Moore, and Christopher Pollett.
\newblock Counting, fanout and the complexity of quantum {ACC}.
\newblock {\em Quantum Information {\&} Computation}, 2(1):35--65, 2002.

\bibitem[HEB04]{Hein+PRA04}
Marc Hein, Jens Eisert, and Hans~J. Briegel.
\newblock Multiparty entanglement in graph states.
\newblock {\em Physical Review A}, 69:062311, 2004.

\bibitem[HS05]{Hoyer+TOC05}
Peter H{\o}yer and Robert Spalek.
\newblock Quantum fan-out is powerful.
\newblock {\em Theory of Computing}, 1(1):81--103, 2005.

\bibitem[KW97]{Kondacs+FOCS97}
Attila Kondacs and John Watrous.
\newblock On the power of quantum finite state automata.
\newblock In {\em Proceedings of the 38th Annual Symposium on Foundations of
  Computer Science}, pages 66--75, 1997.

\bibitem[LNR19]{LeGall+ArXiv18}
Fran{\c c}ois {Le Gall}, Harumichi Nishimura, and Ansis Rosmanis.
\newblock Quantum advantage for the {LOCAL} model in distributed computing.
\newblock In {\em Proceedings of the International Symposium on Theoretical
  Aspects of Computer Science}, pages 49:1--49:14, 2019.

\bibitem[MFF14]{Morimae+PRL14}
Tomoyuki Morimae, Keisuke Fujii, and Joseph~F. Fitzsimons.
\newblock Hardness of classically simulating the one-clean-qubit model.
\newblock {\em Physical Review Letters}, 112:130502, 2014.

\bibitem[NC00]{Nielsen+00}
Michael~A. Nielsen and Isaac~L. Chuang.
\newblock {\em Quantum Computation and Quantum Information}.
\newblock Cambridge University Press, 2000.

\bibitem[RT19]{Raz+18}
Ran Raz and Avishay Tal.
\newblock Oracle separation of {BQP} and {PH}.
\newblock {\em Proceedings of the 43rd {ACM} Symposium on Theory of Computing},
  2019.
\newblock To appear.

\bibitem[TD04]{Terhal+04}
Barbara~M. Terhal and David~P. DiVincenzo.
\newblock Adaptive quantum computation, constant depth quantum circuits and
  {Arthur}-{Merlin} games.
\newblock {\em Quantum Information {\&} Computation}, 4(2):134--145, 2004.

\bibitem[TT16]{Takahashi+CCC16}
Yasuhiro Takahashi and Seiichiro Tani.
\newblock Collapse of the hierarchy of constant-depth exact quantum circuits.
\newblock {\em Computational Complexity}, 25(4):849--881, 2016.

\end{thebibliography}

\newcommand{\etalchar}[1]{$^{#1}$}

\end{document}